\documentclass[a4paper,10pt]{article}

\usepackage{amssymb}
\usepackage{latexsym}
\usepackage{amsmath}
\usepackage{amscd}

\usepackage{color}

\numberwithin{equation}{section}
\newcommand{\bR}{{\mathbb R}}

\newcommand{\bC}{{\mathbb C}}

\newcommand{\kB}{{\mathcal B}}
\newcommand{\kC}{{\mathcal C}}
\newcommand{\kD}{{\mathcal D}}

\newcommand{\kH}{{\mathcal H}}
\newcommand{\kK}{{\mathcal K}}

\newcommand{\kN}{{\mathcal N}}

\newcommand{\kQ}{{\mathcal Q}}
\newcommand{\kS}{{\mathcal S}}

\newcommand{\gotH}{{\mathfrak H}}

\newcommand{\gotK}{{\mathfrak K}}

\newcommand{\gd}{{\delta}}

\newcommand{\gga}{{\gamma}}
\newcommand{\gG}{{\Gamma}}

\newcommand{\gl}{{\lambda}}

\newcommand{\gO}{{\Omega}}
\newcommand{\gP}{{\Pi}}

\newcommand{\gS}{{\Sigma}}
\newcommand{\gs}{{\sigma}}

\newcommand{\gT}{{\Theta}}

\newcommand{\slim}{\,\mbox{\rm s-}\hspace{-2pt} \lim}

\newcommand{\real}{{\Re{\mathrm e\,}}}
\newcommand{\imag}{{\Im{\mathrm m\,}}}
\newcommand{\dom}{{\mathrm{dom\,}}}
\newcommand{\ran}{{\mathrm{ran\,}}}

\newcommand{\spa}{{\mathrm{span}}}
\newcommand{\clospa}{{\mathrm{clospan}}}

\newcommand{\wh}[1]{{\widehat{#1}}}

\newtheorem{thm}{Theorem}[section]
\newtheorem{prop}[thm]{Proposition}
\newtheorem{lem}[thm]{Lemma}
\newtheorem{cor}[thm]{Corollary}

\newtheorem{defn}[thm]{Definition}
\newtheorem{rem}[thm]{Remark}

\newcommand{\ba}{\begin{array}}
\newcommand{\ea}{\end{array}}
\newcommand{\bea}{\begin{eqnarray}}
\newcommand{\eea}{\end{eqnarray}}
\newcommand{\bead}{\begin{eqnarray*}}
\newcommand{\eead}{\end{eqnarray*}}
\newcommand{\be}{\begin{equation}}
\newcommand{\ee}{\end{equation}}
\newcommand{\bed}{\begin{displaymath}}
\newcommand{\eed}{\end{displaymath}}
\newcommand{\bl}{\begin{lem}}
\newcommand{\el}{\end{lem}}
\newcommand{\bp}{\begin{prop}}
\newcommand{\ep}{\end{prop}}
\newcommand{\bt}{\begin{thm}}
\newcommand{\et}{\end{thm}}
\newcommand{\Label}{\label}
\newcommand{\bc}{\begin{cor}}
\newcommand{\ec}{\end{cor}}
\newcommand{\la}{\Label}

\newcommand{\bd}{\begin{defn}}
\newcommand{\ed}{\end{defn}}

\def\cprime{$'$} \def\cprime{$'$} \def\cprime{$'$} \def\cprime{$'$}
  \def\lfhook#1{\setbox0=\hbox{#1}{\ooalign{\hidewidth
  \lower1.5ex\hbox{'}\hidewidth\crcr\unhbox0}}} \def\cprime{$'$}
  \def\cprime{$'$} \def\cprime{$'$}

   \def\sH{{\mathfrak H}}

      \def\dC{{\mathbb C}}

   \def\dN{{\mathbb N}}   
      \def\dR{{\mathbb R}}

   \def\cB{{\mathcal B}}   \def\cC{{\mathcal C}}
\def\cD{{\mathcal D}}      \def\cF{{\mathcal F}}
   \def\cH{{\mathcal H}}   
      
   \def\cN{{\mathcal N}}

\def\mul{{\text{\rm mul\,}}}

\newcommand{\Imag}{\mbox{{\rm Im}}\,}

\newenvironment{proof}%
{\begin{sloppypar}\noindent{\bf Proof.}}%
{\hspace*{\fill}$\square$\end{sloppypar}\bigskip}

\author{Jussi Behrndt\thanks{E-mail: behrndt@math.tu-berlin.de}\\
Technische Universit\"{a}t Berlin\\
Institut f\"{u}r Mathematik\\
MA 6--4, Stra\ss e des 17.\ Juni 136\\
D--10623 Berlin, Germany
\and
Mark M. Malamud\thanks{E-mail: mmm@telenet.dn.ua}\\
Donetsk National University\\
Department of Mathematics\\
Universitetskaya 24\\
83055 Donetsk, Ukraine
\and
Hagen Neidhardt\thanks{E-mail: neidhardt@wias-berlin.de}\\
Weierstra{\ss}-Institut f\"ur\\
Angewandte Analysis und Stochastik\\ 
Mohrenstr. 39\\
D-10117 Berlin, Germany
}

\title{Finite Rank Perturbations, Scattering Matrices and
Inverse Problems}

\date{February 19, 2009}    % Datum fixieren

\begin{document}

\maketitle

Dedicated to the memory of our friend Peter Jonas (18.7.1941 - 18.7.2007)

\begin{abstract}
\noindent
In this paper the
scattering matrix of a scattering system consisting of two selfadjoint operators
with finite dimensional resolvent difference 
is expressed in terms of a matrix Nevanlinna function.
The problem is embedded into an extension theoretic framework and the theory
of boundary triplets and associated Weyl functions for (in general nondensely defined)
symmetric operators is applied. The representation 
results are extended to dissipative scattering systems and 
an explicit solution of an inverse scattering
problem for the Lax-Phillips scattering matrix is presented.\\


\end{abstract}

\section{Introduction}

Let $A$ and $B$ be selfadjoint operators in a Hilbert space $\gotH$ and
assume that the difference of the resolvents
\begin{equation}\label{intdiff}
(B-\lambda)^{-1}-(A-\lambda)^{-1},\qquad\lambda\in\rho(A)\cap\rho(B),
\end{equation}
is a trace class operator. Then it is well known that the wave operators
$W_\pm(B,A)$ exist; they are isometries mapping the absolutely
continuous subspace of $A$ onto that of $B$. The scattering operator
$S_{AB}:=W_+(B,A)^*W_-(B,A)$ of the scattering system $\{A,B\}$
commutes with $A$ and is unitary on the absolutely continuous subspace
of $A$. Therefore $S_{AB}$ is unitarily equivalent to a multiplication operator induced by a family
of unitary operators $\{S_{AB}(\lambda)\}_{\lambda\in\dR}$ in the spectral representation of
the absolutely continuous part of $A$. The family
$\{S_{AB}(\lambda)\}_{\lambda\in\dR}$ is called the scattering matrix
of $\{A,B\}$.

One of the main objectives of this paper is to represent the scattering
matrix of the scattering system
$\{A,B\}$ with the help of an associated Nevanlinna function $M(\cdot)$. We
restrict ourselves to the special case of finite
rank perturbations in resolvent sense, i.e., it is assumed that the
difference of the resolvents in \eqref{intdiff} is a rank $n$
operator, where $n<\infty$. In this case the Nevanlinna function $M(\cdot)$ will be an $n\times n$-matrix
function and it will be shown in Theorem~\ref{scattering}
that the scattering matrix $\{S_{AB}(\lambda)\}_{\lambda\in\dR}$ is given by
\begin{equation}\label{scatint}
S_{AB}(\lambda)=I-2i\sqrt{\imag
(M(\lambda+i0))}\,M(\lambda+i0)^{-1}\,\sqrt{\imag (M(\lambda+i0))}
\end{equation}
for a.e. $\lambda\in\dR$. This representation is a generalization of a
recent result of the authors from \cite{BMN06a} and
an earlier different (unitarily equivalent) expression found by
V.M.~Adamyan and B.S.~Pavlov in \cite{AdamP86}.
The formula \eqref{scatint} is obtained by embedding the scattering
problem into an extension theoretic
framework. More precisely, we consider the (in general nondensely defined)
closed symmetric operator $S=A\cap B$ which has finite equal deficiency
indices $(n,n)$. The adjoint $S^*$ is defined
in the sense of linear relations and a so-called boundary triplet
$\Pi=\{\dC^n,\Gamma_0,\Gamma_1\}$ for $S^*$ is chosen in
such a way that the selfadjoint extensions of $S$ corresponding to the
boundary mappings $\Gamma_0$ and $\Gamma_1$
coincide with $A$ and $B$, respectively. The function $M(\cdot)$ in
\eqref{scatint} is the Weyl function
associated to this boundary triplet -- an abstract analogon of the
classical Titchmarsh-Weyl $m$-function from singular Sturm-Liouville
theory -- and contains the spectral information of
the operator $A$.

Besides selfadjoint scattering systems we also consider so-called
maximal dissipative scattering systems $\{A,B\}$, that is, $A$
is still a selfadjoint operator in $\gotH$ but $B$ is only assumed to be maximal
dissipative, i.e., $\imag (Bf,f)\leq 0$ and the spectrum of $B$ is contained in $\dC_-\cup\dR$.
As above we treat only the case of finite rank perturbations in
resolvent sense. Following \cite{BMN06b,Pav76,Pav96,FoiasN70}
a minimal selfadjoint dilation $L$ of $B$ in the direct sum
$\gotH\oplus L^2(\dR,\dC^n)$ is constructed and a natural
larger selfadjoint scattering system $\{K,L\}$ in $\gotH\oplus
L^2(\dR,\dC^n)$ is considered. From Theorem~\ref{scattering} and Theorem~\ref{III.2} we
obtain a representation of the scattering matrix $\{S_{KL}(\lambda)\}_{\lambda\in\dR}$
which is closely related to the representations
found earlier in \cite{BMN06b}.  We emphasize that the lower right corner of
$\{S_{KL}(\lambda)\}_{\lambda\in\dR}$ in Proposition~\ref{sklprop} can be interpreted
as the Lax-Phillips scattering matrix $\{S^{LP}(\lambda)\}_{\lambda\in\dR}$ of the
Lax-Phillips scattering system $\{L,\cD_-,\cD_+\}$, where the incoming and outcoming subspaces $\cD_-$
and $\cD_+$ are $L^2(\dR_-,\dC^n)$ and $L^2(\dR_+,\dC^n)$, respectively.
This also implies the well known relation
$S^{LP}(\lambda)=\gT_B(\lambda-i0)^*$ between the Lax-Phillips
scattering matrix and the
characteristic function $\gT_B(\cdot)$ of the maximal dissipative operator $B$
found by V.M. Adamyan and D.Z. Arov in \cite{AA65a,AA65b,AA66a,AA66b}.

As an application of our approach on finite rank perturbations and
maximal dissipative scattering systems we prove an
inverse result in Section~\ref{invsec}. Let $W(\cdot)$ be a purely contractive analytic
matrix function on $\dC_+$. Under some mild additional assumptions it is shown 
in Theorem~\ref{invthm} that the limit $\{W(\lambda+i0)\}_{\lambda\in\dR}$ can be regarded as the
Lax-Phillips scattering matrix of a suitable chosen Lax-Phillips scattering system.
We point out that this statement can be
obtained immediately in a more abstract and general form by combining the results of B.~Sz.-Nagy 
and C.~Foia{\lfhook{s}} in \cite[Section VI]{FoiasN70} with the results of  
V.M.~Adamyan and D.Z.~Arov in \cite{AA65a,AA65b,AA66a,AA66b}. However, 
our approach leads to a more explicit solution of the inverse problem, in particular,
we find a maximal dissipative multiplication operator $B$ in an $L^2$-space and
a minimal selfadjoint dilation $L$ of $B$ such that the Lax-Phillips scattering matrix of $\{L,\cD_-,\cD_+\}$
coincides with the limit of the given purely contractive analytic matrix function $W(\cdot)$; 
cf. Corollary~\ref{invcor}.

The paper is organized as follows. In Section~\ref{twosec} we give a brief
introduction in the theory of boundary triplets for 
(in general nondensely defined) closed symmetric operators. In particular, we show how a boundary triplet for the 
intersection $S=A\cap B$ of two selfadjoint operators $A$ and $B$ with a finite dimensional resolvent
difference can be chosen. Section~\ref{threesec}
is devoted to the representation of the scattering matrix for a scattering system 
$\{A,B\}$ with finite rank resolvent difference and in 
Section~\ref{sec4} the results are extended to the case where the operator $B$
is only maximal dissipative. With the help of these results we propose a solution for the 
inverse scattering problem in Section~\ref{invsec}. For the convenience of the
reader we add an Appendix on direct integrals, spectral representations
and scattering matrices.

{\bf Notation.} The Hilbert spaces in this paper are usually denoted by $\gotH$, $\gotK$ and $\kH$; they
are all assumed to be separable. 
The symbols $\spa\{\cdot\}$ and $\clospa\{\cdot\}$ are used for 
the linear span and closed linear span, respectively, of a set. 
The algebra of everywhere defined bounded linear operators on a Hilbert space $\gotH$ with values in a Hilbert space $\gotK$ is 
denoted by $[\gotH,\gotK]$; we write $[\gotH]$ if $\gotK=\gotH$. By $\cF_n(\gotH)$ we denote the subset of $[\gotH]$ that consists of linear operators
with range of dimension $n\in\dN$. The absolutely continuous part of a selfadjoint
operator $A$ in $\gotH$ is denoted by $A^{ac}$, the corresponding subspace by
$\gotH^{ac}(A)$. The symbols $\rho(\cdot)$, $\gs(\cdot)$, 
$\gs_p(\cdot)$, $\gs_c(\cdot)$, $\gs_r(\cdot)$
stand for the resolvent set, the spectrum, the point, continuous and
residual spectrum, respectively. By $E(\cdot)$ and
$\gS(\cdot)$ we denote operator-valued measures defined on the algebra of
Borel sets $\kB(\dR)$ of the real axis $\dR$. Usually, the symbol
$E(\cdot)$ is reserved for orthogonal operator-valued measures.

\section{Selfadjoint and maximal dissipative 
extensions of nondensely defined symmetric operators}\label{twosec}

\subsection{Linear relations}\label{linrel}

Let $(\gotH,(\cdot,\cdot))$ be a separable Hilbert space. A
(closed) linear relation $T$ in $\gotH$ is a (closed) linear
subspace of the Cartesian product space $\gotH\times\gotH$. The set of closed linear relations
in $\gotH$ is denoted by $\widetilde\cC(\gotH)$.
Linear operators in $\gotH$ will always be identified with linear 
relations via their graphs. The elements
of a linear relation $T$ are pairs denoted by $\wh f=\{f,f^\prime\}\in T$,
$f,f^\prime\in\sH$, and the {\it domain},
{\it kernel}, {\it range}, and the {\it multi-valued part} of $T$ are
defined as
\begin{equation*}
 \begin{split}
\dom(T)&=\{\,f\in\sH:\{f,f^\prime\}\in T\,\},\quad\,\,\,\,
\ker(T)=\{\,f\in\sH:\{f,0\}\in T\,\},\\
\ran (T)&=\{\,f^\prime\in\sH:\{f,f^\prime\}\in T\,\},\quad
\mul(T)=\{\, f^\prime\in\sH:\{0,f^\prime\}\in T\,\},
 \end{split}
\end{equation*}
respectively. Note that $T$ is an operator if and only if $\mul (T)=\{0\}$. 
A point $\lambda$ belongs to the {\it resolvent set} $\rho(T)$ of a
closed linear relation $T$ if $(T-\lambda)^{-1}$ is 
an everywhere defined bounded operator in $\gotH$.
The {\it spectrum} $\sigma(T)$ of $T$ is the complement of $\rho(T)$ in $\dC$.

A linear relation $T$ in $\sH$ is called {\it dissipative} if
$\Imag(f^\prime,f)\leq 0$  holds for all
$\{f,f^\prime\}\in T$. A dissipative relation $T$ 
is said to be {\it maximal dissipative} if there exists no proper dissipative
extension of $T$ in $\sH$. It can be shown that a dissipative relation $T$ is maximal
dissipative if and only if $\dC_+\subset\rho(T)$ holds.

The {\it adjoint} $T^*$ of a linear relation $T$ in $\gotH$ is a
closed linear relation in $\sH$ defined by
\begin{equation}\label{adjrel}
T^*:=\bigl\{\, \{g,g^\prime\}:(f^\prime,g)=(f,g^\prime)\,\,\text{for
all}\,\,\{f,f^\prime\} \in T \,\bigr\}.
\end{equation}
Observe that this definition extends the usual definition of the
adjoint operator and that $\mul(T^*)=(\dom(T))^\bot$ holds. 
In particular, $T^*$ is an operator if and only if $T$ is densely defined.
A linear relation $T$ in $\sH$ is called
\textit{symmetric} (\textit{selfadjoint}) if $T\subset T^*$
($T=T^*$, respectively). It follows from the polarization identity
that $T$ is symmetric if and only if $(f^\prime,f)\in\dR$ for all
$\{f,f^\prime\}\in T$.

A (possibly nondensely defined) symmetric operator $S$ in $\sH$ is said to be {\it simple}
if there is no nontrivial subspace in $\sH$ which reduces $S$ to a selfadjoint operator.
It is well known that every symmetric operator $S$ can be written as the direct orthogonal sum
$\widehat S\oplus S_s$ of a simple symmetric operator $\widehat S$ in
the Hilbert space
\begin{equation}\label{simpl}
\widehat\sH=\clospa\bigl\{\ker(S^*-\lambda):\lambda\in\dC\backslash\dR\bigr\}
\end{equation}
and a selfadjoint operator $S_s$ in $\sH\ominus\widehat \sH$.

\subsection{Boundary triplets for nondensely defined symmetric operators}

Let in the following $S$ be a (not necessarily densely defined) closed symmetric operator in
the separable Hilbert space $\gotH$ with equal deficiency indices
\begin{equation*}
n_\pm(S)=\dim\bigl(\ran(S\pm i)^\bot\bigr)=\dim\bigl(\ker(S^*\mp i)\bigr)\leq\infty.
\end{equation*}
If $\dom(S)$ is not dense in $\gotH$ the adjoint $S^*$ exists only in the sense of linear relations 
and is defined as in \eqref{adjrel}. Therefore, if $S$ is not densely defined 
the closed extensions $S^\prime \subset S^*$ 
of $S$ in $\gotH$ may have nontrivial multi-valued parts. 
However, the operator $S$ admits also
closed extensions in $\gotH$ which are operators. We will use the 
concept of boundary triplets for the description of the
closed extensions $S^\prime \subset S^*$ of $S$ in $\gotH$; 
see, e.g., \cite{BGP08,DM87a,DM91,DM95,GG91,Mal92}.
This concept also provides a convenient criterion to check 
whether $S^\prime$ is an operator or not; cf. \eqref{check}.
\begin{defn}
{\rm
A triplet $\gP = \{\kH,\gG_0,\gG_1\}$ is called a {\rm boundary triplet} for $S^*$ if $\kH$ is a
Hilbert space and $\Gamma_0,\Gamma_1:S^* \rightarrow\kH$ are
linear mappings such that the abstract Green's identity 
\begin{equation*}
(f^\prime,g) - (f,g^\prime) = (\gG_1\hat f,\gG_0\hat g) -
(\gG_0\hat f,\gG_1\hat g)
\end{equation*}
holds for all $\widehat f=\{f,f^\prime\}$, $\widehat g=\{g,g^\prime\}\in S^*$ and the mapping
 $\gG:=(\Gamma_0,\Gamma_1)^\top: S^*\rightarrow \kH \oplus \kH$ is
 surjective.
}
\end{defn}

We refer to \cite{DM91,DM95,GG91,Mal92} for a detailed study of
boundary triplets and recall only some important facts. First of all
a boundary triplet $\Pi=\{\kH,\gG_0,\gG_1\}$ for $S^*$ exists (but is not unique) since
the deficiency indices $n_\pm(S)$ of $S$ are assumed to be equal.
Then $n_\pm(S) = \dim\kH$ holds. A standard construction of a boundary triplet
will be given in the proof of Proposition~\ref{pertupropd}. 

Let  $\Pi=\{\kH,\gG_0,\gG_1\}$ be a  
boundary triplet for $S^*$ and let ${\rm Ext}(S)$ be the set of all closed extensions $S^\prime\subset S^*$ of $S$. 
Then $S=\ker(\Gamma)$ and the mapping
\begin{equation}\label{bij}
\Theta\mapsto S_\Theta:= \Gamma^{-1}\Theta=
\bigl\{\widehat f\in S^*: \{\Gamma_0\widehat f,\Gamma_1\widehat f\}\in\Theta\bigr\}
\end{equation}
establishes a bijective correspondence between the set 
$\widetilde\kC(\kH)$ of closed linear relations in $\gotH$ 
and the set of closed extensions $S_\Theta\in {\rm Ext}(S)$ of $S$. 
We note that the right-hand side of \eqref{bij} can also be 
written as $\ker(\Gamma_1-\Theta\Gamma_0)$ where the sum and 
product is interpreted in the sense of linear relations.
Since $(S_\Theta)^*=  S_{\Theta^*}$ holds for every 
$\Theta\in\widetilde\kC(\kH)$ it follows that
$S_\Theta$ is  symmetric (selfadjoint) if and only if 
$\Theta$ is symmetric (selfadjoint, respectively). In particular, 
the extensions $A_0:=\ker(\Gamma_0)$ and $A_1:=\ker(\Gamma_1)$ are
selfadjoint. The selfadjoint operator or relation $A_0$ 
will often play the role of a fixed selfadjoint 
extension of $S$ in $\gotH$. Furthermore, an extension 
$S_\Theta\in {\rm Ext}(S)$ is dissipative (maximal dissipative) if 
and only if $\Theta$ is dissipative (maximal dissipative,
respectively). We note that $S_\Theta$ in \eqref{bij} is an operator
if and only if
\begin{equation}\label{check}
 \Theta\cap\bigl\{\{\Gamma_0\widehat f,\Gamma_1\widehat f\}:
\widehat f=\{0,f^\prime\}\in S^*\bigr\}=\{0\}.
\end{equation}
The following proposition is a consequence of 
the basic properties of boundary triplets and results from \cite{DM91,DM95,Mal92}. 
Since it plays an important role
in this paper we give a complete proof for the convenience of the
reader. 
We also note that the statement remains true if $A$ and $B$
are linear relations instead of operators. Recall that $\cF_n(\gotH)$, $n\in\dN$, is the set of finite dimensional
operators in $\gotH$ with ranges of dimension $n$, i.e., 
\begin{equation*}
\cF_n(\gotH)=\bigl\{T\in[\gotH]:\dim(\ran(T))=n\bigr\}.
\end{equation*}
\begin{prop}\label{pertupropd}
Let $A$ be a selfadjoint operator and let $B$ be a 
maximal dissipative operator in $\gotH$. Assume that
\begin{equation*}
(B-\lambda)^{-1}-(A-\lambda)^{-1}\in\cF_n(\gotH)
\end{equation*}
holds for some (and hence for all) $\lambda\in\dC_+$. 
Then the closed symmetric operator $S:=A\cap B$ 
has finite deficiency indices $(n,n)$ in $\gotH$ and there
exists a boundary triplet $\gP = \{\dC^n,\Gamma_0,\Gamma_1\}$ for $S^*$ 
and a dissipative $n\times n$-matrix $D$
such that $A=\ker(\Gamma_0)$ and $B=\ker(\Gamma_1-D\Gamma_0)$ holds. 
\end{prop}
\begin{proof}
Let $\lambda_0\in\rho(A)\cap\rho(B)$ and let $n\in\dN$, 
$\{e_1,\dots,e_n\}$ and $\{f_1,\dots, f_n\}$
be linearly independent vectors such that
\begin{equation}\label{pertu}
(B-\lambda_0)^{-1}-(A-\lambda_0)^{-1}=\sum_{i=1}^n (\cdot,e_i)f_i.
\end{equation}
The operator $S=A\cap B$, that is,
\begin{equation*}
Sf=Af=Bf,\quad\dom S=\bigl\{f\in\dom A\cap\dom B:Af=Bf\bigr\},
\end{equation*}
is a (in general non-densely defined) symmetric operator in 
$\gotH$ and it is easy to check that
\begin{equation}\label{ressab}
(S-\lambda_0)^{-1}=(A-\lambda_0)^{-1}\cap(B-\lambda_0)^{-1}
\end{equation}
holds. The intersection in \eqref{ressab} is understood in the sense of linear relations. 
Hence \eqref{pertu} and \eqref{ressab} imply
$\dim(\ran(A-\gl_0)^{-1})/\ran(S-\gl_0)^{-1}))=n$. Therefore $\dim(A/S)=n$ and
$S$ has deficiency indices $(n,n)$. Note that $(S-\lambda_0)^{-1}$ is defined
on the subspace $\gotH\ominus\spa\{e_1,\dots,e_n\}$ which has codimension $n$ in $\gotH$.

It is not difficult to verify that $S^*$ coincides 
with the direct sum of the subspaces $A$ and 
\begin{equation*}
\widehat\cN_{\lambda_0}=
\bigl\{\{f_{\lambda_0}, \lambda_0 f_{\lambda_0}\}:\,f_{\lambda_0}\in\cN_{\lambda_0}
=\ker(S^*-\lambda_0)\bigr\}.
\end{equation*}
Let us decompose the elements $\widehat f\in S^*$ accordingly, i.e.,
\begin{equation}\label{decfhat}
\widehat f=\{ f, f^\prime\}=
\bigl\{f_A + f_{\lambda_0} , A f_A+
\lambda_0 f_{\lambda_0} \bigr\}     ,\quad f_A\in\dom A,\, f_{\lambda_0}\in\cN_{\lambda_0},
\end{equation}
and denote by $P_0$ the orthogonal projection onto the closed subspace $\cN_{\lambda_0}$.
Then $\Pi=\{\cN_{\lambda_0},\Gamma_0,\Gamma_1\}$, where
\begin{equation*}
\Gamma_0\widehat f:=f_{\lambda_0}\quad\text{and}\quad
\Gamma_1\widehat f:=P_0\bigl((A-\bar\lambda_0)f_A+\lambda_0 f_{\lambda_0}\bigr),
\end{equation*}
$\widehat f\in S^*$, is a boundary triplet 
with $A=A_0:=\ker(\Gamma_0)$. In fact, for $\widehat f$ as in \eqref{decfhat}
and $\widehat g=\{g,g^\prime\}=\{g_A + g_{\lambda_0} , A g_A+
\lambda_0 g_{\lambda_0}\}$ we obtain from $(Af_A,g_A)=(f_A,Ag_A)$ that
\begin{equation*}
\begin{split}
(f^\prime,g)-(f,g^\prime)&=
\bigl((A-\bar\lambda_0)f_A+\lambda_0 f_{\lambda_0},g_{\lambda_0}\bigr)
-\bigl(f_{\lambda_0},(A-\bar\lambda_0)g_A+\lambda_0 g_{\lambda_0}\bigr)\\
&=(\Gamma_1\widehat f,\Gamma_0\widehat g)-(\Gamma_0\widehat f,\Gamma_1\widehat g)
\end{split}
\end{equation*}
holds. The surjectivity of the mapping 
$\Gamma=(\Gamma_0,\Gamma_1)^\top:S^*\rightarrow\cN_{\lambda_0}
\oplus\cN_{\lambda_0}$ follows from $\bar\lambda_0\in\rho(A)$ 
since for $x,x^\prime\in\cN_{\lambda_0}$
we can choose $f_A\in\dom A$ such that $(A-\bar\lambda_0)f_A=x^\prime-\lambda_0 x$ holds.
Then obviously $\widehat f:=\{f_A+x , Af_A+\lambda_0 x\}$
satisfies $\Gamma\widehat f=(x,x^\prime)^\top$. 
Moreover, from the definition of $\Gamma_0$ we immediately
obtain that the extension $A_0=\ker(\Gamma_0)$ coincides 
with the operator $A$. As the deficiency indices of $S$
are $(n,n)$ we can identify $\cN_{\lambda_0}$ with $\dC^n$.

Since $B$ is a maximal dissipative extension of the symmetric operator $S$, $B\in {\rm Ext}(S)$.
Hence $B\subset\dom(\Gamma)$ and the linear relation
\bed
D:=\Gamma B=\bigl\{\{\Gamma_0\widehat f, \Gamma_1\widehat f\}:
\widehat f=\{f, Bf\}\in B \bigr\}
\eed
is maximal dissipative in $\dC^n$ and $B$ coincides with 
the maximal dissipative extension $S_D$ via~\eqref{bij}.
We claim that $D$ is a matrix, i.e., $\mul (D)=\{0\}$.
In fact, assume that $D$ is multi-valued, that is, there exists $\widehat f=\{f,Bf\}\in B$ such that
$\{0,\Gamma_1\widehat f\}\in D$ with $\Gamma_1\widehat f\not=0$.
In particular, $\Gamma_0\widehat f=0$, i.e., $\widehat f\in
A_0=A$ and therefore $\widehat f\in A\cap
B=S=\ker(\Gamma_0,\Gamma_1)^\top$,  however, this is a contradiction.
Thus $D$ is a dissipative $n\times n$-matrix and it follows
from \eqref{bij} that $B=\ker(\Gamma_1-D\Gamma_0)$ holds.
\end{proof}

\subsection{Weyl functions and Krein's formula}\label{btrips2}

Again let $S$ be a (in general nondensely defined) closed symmetric operator in 
$\gotH$ with equal deficiency indices as in the previous section.
If $\lambda\in\dC$ is a point of regular type of $S$, i.e.,
$(S-\lambda)^{-1}$ is a bounded operator, we denote the {\it defect subspace}
of $S$ at $\lambda$ by $\kN_\gl=\ker(S^*-\gl)$ and we agree to write
\begin{equation*}
\widehat\cN_\lambda=\bigl\{\{f, \lambda f\}:f\in\cN_\lambda\bigr\}\subset S^*.
\end{equation*}
Let $\Pi=\{\cH,\Gamma_0,\Gamma_1\}$ be a boundary triplet 
for $S^*$ and let $A_0=\ker(\Gamma_0)$ be the fixed selfadjoint 
extension of $S$.
Recall that for every $\lambda\in\rho(A_0)$ the 
relation $S^*$ is the direct sum of the selfadjoint relation $A_0$ and $\widehat\cN_\lambda$ and 
denote by $\pi_1$ the orthogonal projection onto the first component of
$\gotH\oplus\gotH$. The operator valued functions
\bed
\gamma(\cdot):\rho(A_0)\rightarrow  [\kH,\gotH],\quad \lambda\mapsto
\gamma(\lambda)=\pi_1\bigl(\Gamma_0\!\upharpoonright\!\widehat\cN_\lambda\bigr)^{-1}
\eed
and
\bed
M(\cdot):\rho(A_0)\rightarrow  [\kH],\quad \lambda\mapsto
M(\lambda)=\Gamma_1\bigl(\Gamma_0\!\upharpoonright\!\widehat\cN_\lambda\bigr)^{-1}
\eed
are called the {\it $\gamma$-field} and the {\it Weyl function}, respectively,
corresponding to the boundary triplet 
$\Pi=\{\cH,\Gamma_0,\Gamma_1\}$; see, e.g., \cite{DM87a,DM91,DM95,Mal92}.
It can be shown that both  $\gamma(\cdot)$ and $M(\cdot)$ are
holomorphic on $\rho(A_0)$ and that the identities
\begin{equation}\label{gammamu}
\gamma(\mu)=\bigl(I+(\mu-\gl)(A_0-\mu)^{-1}\bigr)\gamma(\gl),
\qquad \gl,\mu\in\rho(A_0),
\end{equation}
and
\begin{equation}\label{mlambda}
M(\gl)-M(\mu)^*=(\gl-\bar\mu)\gamma(\mu)^*\gamma(\gl),
\qquad \gl,\mu\in\rho(A_0),
\end{equation}
are valid; see \cite{DM91,Mal92}.
The identity \eqref{mlambda} yields that $M(\cdot)$ is a  $[\kH]$-valued {\it
Nevanlinna function}, that is, $M(\cdot)$ is holomorphic on $\dC\backslash\dR$, 
$\imag(M(\gl))$ is a nonnegative operator for all $\gl\in\dC_+$ and $M(\gl)=M(\bar\gl)^*$ 
holds for all $\lambda\in\dC\backslash\dR$. 
Moreover, it follows from \eqref{mlambda}
that $0\in \rho(\imag(M(\gl)))$ for
all $\gl\in\dC\backslash\bR$ and, in particular,
\begin{equation}\label{imm}
\frac{\imag (M(\gl))}{\imag (\lambda)}= \gamma(\gl)^*\gamma(\gl),\qquad\lambda\in\dC\backslash\dR.
\end{equation}  
The following inverse result is essentially a consequence
of \cite{LangT77}, see also \cite{DM95,Mal92}.
\begin{thm}\label{nevweyl}
Let $M:\dC\backslash\dR\rightarrow [\cH]$ be a Nevanlinna function such that
$0\in\rho(\imag(M(\gl)))$ for some (and hence for all) 
$\lambda\in\dC\backslash\dR$ and assume that the condition
\begin{equation}\label{2.9}
\lim_{\eta\rightarrow+\infty}\frac{1}{\eta}(M(i\eta)h,h)=0
\end{equation}
holds for all $h\in\cH$. Then there exists a separable Hilbert 
space $\gotH$, a closed simple symmetric operator $S$ in $\gotH$  and a boundary triplet 
$\gP = \{\cH,\Gamma_0,\Gamma_1\}$ for the adjoint relation $S^*$ such that $A_0=\ker(\Gamma_0)$ is a selfadjoint operator and the
Weyl function of $\gP$ coincides with 
$M(\cdot)$ on $\dC\backslash\dR$. The symmetric operator $S$ is densely defined
if and only if the conditions \eqref{2.9} and 
\bed
\lim_{\eta\rightarrow+\infty}\eta\,\imag(M(i\eta)h,h)=\infty,\qquad h\in\cH,\,h\not=0,
\eed
are satisfied.
\end{thm}

The spectrum and the resolvent set of closed extensions
in ${\rm Ext}(S)$ can be described with the help of the
Weyl function. More precisely, if $S_\Theta\in{\rm Ext}(S)$ is the extension
corresponding to $\Theta\in\widetilde\kC(\kH)$ via \eqref{bij}, then
a point $\gl\in\rho(A_0)$ belongs to
$\rho(S_\Theta)$ ($\sigma_i(S_\Theta)$, $i=p,c,r$) if and only if $0\in\rho(\Theta-M(\gl))$ 
($0\in\sigma_i(\Theta-M(\gl))$, $i=p,c,r$, respectively). Moreover, for
$\gl\in\rho(A_0)\cap\rho(S_\Theta)$ the well-known resolvent formula
\begin{equation}\label{2.8}
(S_\Theta - \gl)^{-1} = (A_0 - \gl)^{-1} + \gga(\gl)\bigl(\Theta -
M(\gl)\bigr)^{-1}\gga(\bar{\gl})^*
\end{equation}
holds, see \cite{DM91,Mal92}. Formula \eqref{2.8} and Proposition~\ref{pertupropd} 
imply the following statement which will be used in Section~\ref{sec4}.
\begin{cor}\label{pertupropcor}
Let $A$ be a selfadjoint operator and let $B$ be a 
maximal dissipative operator in $\gotH$ such that 
\begin{equation*}
(B-\lambda)^{-1}-(A-\lambda)^{-1}\in\cF_n(\gotH)
\end{equation*}
holds for some (and hence for all) $\lambda\in\dC_+$. 
Let $\gP = \{\dC^n,\Gamma_0,\Gamma_1\}$ be the boundary triplet from Proposition~\ref{pertupropd}
such that $A=\ker(\Gamma_0)$ and $B=\ker(\Gamma_1-D\Gamma_0)$ 
holds with some dissipative $n\times n$-matrix $D$
and denote the $\gamma$-field and the Weyl function of $\gP$
by $\gamma(\cdot)$ and $M(\cdot)$, respectively. Then
\begin{equation}\label{resab}
(B-\lambda)^{-1}-(A-\lambda)^{-1}=\gamma(\lambda)\bigl(D-M(\lambda)\bigr)^{-1}\gamma(\bar\lambda)^*
\end{equation}
holds for all $\lambda\in\rho(B)\cap\rho(A)$.
\end{cor}

If the maximal dissipative operator $B$ in 
Proposition~\ref{pertupropd} and Corollary~\ref{pertupropcor} is even selfadjoint
the representation of the resolvent difference in \eqref{resab} can be further simplified.
\begin{cor}\label{pertuprop}
Let $A$ and $B$ be selfadjoint operators in $\gotH$ such that 
\begin{equation*}
(B-\lambda)^{-1}-(A-\lambda)^{-1}\in\cF_n(\gotH)
\end{equation*}
holds for some (and hence for all) $\lambda\in\dC\backslash\dR$. 
Then the closed symmetric operator $S=A\cap B$ 
has finite deficiency indices $(n,n)$ in $\gotH$ and there
exists
a boundary triplet $\gP = \{\dC^n,\Gamma_0,\Gamma_1\}$ for $S^*$ 
such that $A=\ker(\Gamma_0)$ and $B=\ker(\Gamma_1)$ holds. Moreover, if $\gamma(\cdot)$ and $M(\cdot)$
denote the $\gamma$-field and Weyl function of $\gP$, then
\bed
(B-\lambda)^{-1}-(A-\lambda)^{-1}=-\gamma(\lambda)M(\lambda)^{-1}\gamma(\bar\lambda)^*
\eed
holds for all $\lambda\in\rho(B)\cap\rho(A)$.
\end{cor}
\begin{proof}
According to Proposition~\ref{pertupropd} there 
is a boundary triplet $\gP^\prime = \{\dC^n,\Gamma_0^\prime,\Gamma_1^\prime\}$
for $S^*$ such that $A=\ker(\Gamma_0^\prime)$ and 
$B=\ker(\Gamma_1^\prime-D\Gamma_0^\prime)$. Here the dissipative
matrix $D$ is even symmetric since $B$ is selfadjoint. A simple calculation 
shows that $\gP = \{\dC^n,\Gamma_0,\Gamma_1\}$, where
\begin{equation*}
\Gamma_0:=\Gamma_0^\prime\quad\text{and}\quad\Gamma_1:=\Gamma_1-D\Gamma_0, 
\end{equation*}
is also a boundary triplet for $S^*$. If $M(\cdot)$ is the Weyl 
function corresponding to the boundary triplet $\Pi^\prime$, 
then $\lambda\mapsto M(\lambda)-D$ is the Weyl function 
corresponding to the boundary triplet $\Pi$. This together 
with Proposition~\ref{pertupropd} and Corollary~\ref{pertupropcor} implies the statement.
\end{proof}

\section{A representation of the scattering matrix}\label{threesec}

In this section we consider scattering systems $\{A,B\}$ consisting of
two selfadjoint operators $A$ and $B$ in a separable Hilbert space $\gotH$
and we assume that the difference of the resolvents of $A$ and $B$ is a
finite rank operator, that is, for some $n\in\dN$ we have
\begin{equation}\label{finite}
(B-\lambda)^{-1}-(A-\lambda)^{-1}\in\cF_n(\gotH)
\end{equation}
for one (and hence for all) $\lambda\in\rho(A)\cap\rho(B)$.
Then the {\it wave operators}
\bed
W_\pm(B,A) := \slim_{t\to\pm\infty}e^{itB}e^{-itA}P^{ac}(A),
\eed
exist and are complete, where $P^{ac}(A)$ denotes the orthogonal
projection onto the absolutely continuous subspace $\gotH^{ac}(A)$
of $A$. Completeness means that the ranges of
$W_\pm(B,A)$ coincide with the absolutely continuous
subspace $\gotH^{ac}(B)$ of $B$; cf. \cite{BW83,Ka76,Wei03,Yaf92}.
The {\it scattering operator} $S_{AB}$ of the {\it scattering system}
$\{A,B\}$ is defined by
\bed
S_{AB}:= W_+(B,A)^*W_-(B,A).
\eed
Since the scattering operator commutes with $A$ it follows
that $S_{AB}$ is unitarily equivalent to a multiplication operator
induced by a family $\{S_{AB}(\lambda)\}_{\lambda\in\dR}$ of unitary operators in
a spectral representation of
$A^{ac}:=A\upharpoonright \dom(A)\cap\gotH^{ac}(A)$.
The aim of this section is to generalize a representation result
of this so-called
{\it scattering matrix} $\{S_{AB}(\lambda)\}_{\gl \in \bR}$ from \cite{BMN06a}.

According to \eqref{finite} and Corollary~\ref{pertuprop} the (possibly nondensely defined)
closed symmetric operator $S=A\cap B$ has deficiency indices $(n,n)$ and there exists a boundary triplet
$\gP = \{\dC^n,\gG_0,\gG_1\}$ for $S^*$ such that $A =\ker(\gG_0)$ and $B =\ker(\gG_1)$.
The Weyl function $M(\cdot)$ corresponding to the boundary triplet $\Pi$
is a $[\dC^n]$-valued Nevanlinna function. Therefore the limit
\begin{equation}\label{mlim}
M(\gl) :=M(\gl+i0)=\lim_{\varepsilon\rightarrow +0} M(\lambda+i\varepsilon)
\end{equation}
from the upper half-plane $\dC_+$ exists for a.e. $\lambda\in\dR$; see
\cite{Don74,Gar81}. 
As $\imag(M(\lambda))$ is uniformly positive (uniformly negative) for all $\lambda\in\dC_+$ ($\lambda\in\dC_-$, respectively) 
the inverses 
$M(\lambda)^{-1}$ exist for all $\lambda\in\dC\backslash\dR$ and $-M(\cdot)^{-1}$ is also
a $[\dC^n]$-valued Nevanlinna function. Hence it follows that
the limit $\lim_{\varepsilon\rightarrow 0+}M(\gl+i\varepsilon)^{-1}$ exists for a.e. $\gl \in \bR$
and coincides with the inverse of $M(\gl)$ in \eqref{mlim} for a.e. $\gl \in \bR$. 

In the following theorem we find a representation of the scattering matrix $\{S_{AB}(\lambda)\}_{\lambda\in\dR}$
of the scattering system $\{A,B\}$ in the direct integral
$L^2(\bR,d\gl,\kH_\gl)$, where 
\begin{equation}\label{hlambda}
\kH_\gl := \ran(\imag(M(\gl+i0))\quad\text{for a.e.}\quad \gl \in \bR,
\end{equation}
cf. Appendix~A. We will formulate and prove our result first for the case of 
a simple symmetric operator $S=A\cap B$ and discuss the general case afterwards in Theorem~\ref{III.2}.
For the special case that the simple symmetric operator $S=A\cap B$ is densely defined
Theorem~\ref{scattering} reduces to \cite[Theorem 3.8]{BMN06a}. We remark that the proof of
Theorem~\ref{scattering} differs from the proof of \cite[Theorem 3.8]{BMN06a}. Here we make use
of the abstract representation result Theorem~\ref{A.II}. 
\begin{thm}\label{scattering}
Let $A$ and $B$ be selfadjoint operators in $\gotH$ such that \eqref{finite}
is satisfied, suppose that the symmetric operator $S=A\cap B$ is simple and let 
$\gP = \{\dC^n,\Gamma_0,\Gamma_1\}$ be a boundary triplet
for $S^*$ such that $A=\ker(\Gamma_0)$ and $B=\ker(\Gamma_1)$; cf. Corollary~\ref{pertuprop}.
Let $M(\cdot)$ be the corresponding Weyl function and define the spaces $\kH_\gl$ for a.e. $\lambda\in\dR$ 
as in \eqref{hlambda}.

Then 
$L^2(\bR,d\gl,\kH_\gl)$ performs a spectral representation of
$A^{ac}$ such that the scattering matrix $\{S_{AB}(\gl)\}_{\gl
\in \bR}$ of the scattering system $\{A,B\}$ admits the
representation
\begin{equation}\label{scatformula}
S_{AB}(\gl) = I_{\kH_\gl} -
2i\sqrt{\imag(M(\gl))}\,M(\gl)^{-1}\sqrt{\imag(M(\gl))}\in [\cH_\lambda]
\end{equation}
for a.e. $\gl \in \bR$, where $M(\gl)= M(\gl + i0)$.
\end{thm}
\begin{proof}
In order to verify the representation \eqref{scatformula} of the scattering matrix 
$\{S_{AB}(\lambda)\}_{\lambda\in\dR}$
we will make use of Theorem~\ref{A.II}. For this let us first rewrite the difference of the resolvents
$(B -i)^{-1}$ and $(A - i)^{-1}$ as in \eqref{A.5}. 
According to Corollary~\ref{pertuprop} we have
\begin{equation}\label{diffi}
(B -i)^{-1} - (A - i)^{-1} = - \gga(i)M(i)^{-1}\gga(-i)^*.
\end{equation}
Using \eqref{gammamu} we find
\bed
(B -i)^{-1} - (A - i)^{-1} = -(A+i)(A-i)^{-1}\gga(-i)M(i)^{-1}\gga(-i)^*.
\eed
and hence the representation \eqref{A.5} follows if we set
\begin{equation}\label{settings}
\phi(t):=\frac{t+i}{t-i},\quad t\in\dR,\qquad C=\gamma(-i)\quad\text{and}\quad G=-M(i)^{-1}.
\end{equation}
Moreover, since $S$ is simple it follows from \eqref{simpl} that 
\bed
\gotH = \clospa\bigl\{\ker(S^*-\lambda):\lambda\in\dC\backslash\dR\bigr\}
\eed
holds. As $\ran C=\ran\gamma(-i)=\ker(S^*+i)$ one concludes in the same way as in the proof of 
\cite[Lemma 3.4]{BMN06a} that the condition 
\begin{equation*}
\gotH^{ac}(A)=\clospa\bigl\{E^{ac}_A(\delta)\ran(C):\delta\in\cB(\dR)\bigr\}
\end{equation*}
is satisfied.

Next we express the $[\dC^n]$-valued function
\bed
\lambda\mapsto K(\lambda)=\frac{d}{d\lambda}\,C^*E_A((-\infty,\lambda))C
\eed
and its square root $\lambda\mapsto\sqrt{K(\lambda)}$ in terms of the
Weyl function $M(\cdot)$ for a.e. $\lambda\in\dR$. We have
\begin{equation}\label{k12}
\begin{split}
K(\gl) & =  \lim_{\varepsilon\to+0}\frac{1}{2\pi i}
\gga(-i)^*\left((A- \gl - i\varepsilon)^{-1} - (A- \gl + i\varepsilon)^{-1}\right)\gga(-i)\\
& =  \lim_{\varepsilon\to+0}\frac{\varepsilon}{\pi}\gga(-i)^*
(A- \gl - i\varepsilon)^{-1} (A- \gl + i\varepsilon)^{-1}\gga(-i)
\end{split}
\end{equation}
and on the other hand by \eqref{imm} 
\begin{equation*}
\imag(M(\gl + i\varepsilon)) = \varepsilon\gga(\gl+i\varepsilon)^*\gga(\gl+i\varepsilon). 
\end{equation*}
Inserting $\gga(\gl+i\varepsilon)=(I+(\lambda+i\varepsilon+i)(A-\lambda-i\varepsilon)^{-1}\gga(-i))$ (cf. \eqref{gammamu})
we obtain
\begin{equation}\label{k13}
\imag(M(\gl + i\varepsilon))
=\varepsilon\gga(-i)^*(I+A^2)(A- \gl - i\varepsilon)^{-1} (A- \gl + i\varepsilon)^{-1}\gga(-i)
\end{equation}
and by comparing \eqref{k12} and \eqref{k13} we find
\begin{equation}\label{mk}
\imag(M(\gl)) =\lim_{\varepsilon\rightarrow 0+} \imag(M(\gl + i\varepsilon))=\pi(1 + \gl^2)K(\gl)
\end{equation}
for a.e. $\lambda\in\dR$. In particular, the finite-dimensional subspaces $\ran (K(\lambda))$ in 
Theorem~\ref{A.II} coincide with the spaces $\cH_\lambda=\ran (\imag(M(\lambda)))$ for a.e. $\lambda\in\dR$
and therefore
$L^2(\bR,d\gl,\kH_\gl)$ is a spectral representation of $A^{ac}$
and the scattering matrix $\{S_{AB}(\gl)\}_{\gl \in \bR}$ admits the
representation \eqref{A.3}. Inserting the square root $\sqrt{K(\lambda)}$ from \eqref{mk} into \eqref{A.3}
we find
\begin{equation}\label{scatab1}
S_{AB}(\lambda)=I_{\cH_\lambda}+2  i(1+\lambda^2)\sqrt{\imag(M(\lambda))}Z(\lambda)\sqrt{\imag(M(\lambda))}
\end{equation}
and it remains to compute 
\begin{equation}\label{zzz}
Z(\lambda)=\frac{1}{\lambda+i}Q^*Q+\frac{\phi(\lambda)}{(\lambda+i)^2}G+
\lim_{\varepsilon\rightarrow 0+}Q^*(B-\lambda-i\varepsilon)^{-1}Q,
\end{equation}
where $Q=\phi(A)CG=-\gamma(i)M(i)^{-1}$, cf. \eqref{A.2}, \eqref{settings} and \eqref{diffi}.

It follows from \cite[Lemma 3.2]{BMN06a} that
\begin{equation}\label{asdfg}
Q^*(B-\lambda-i0)^{-1}Q = \frac{1}{1+\lambda^2}\left(M(i)^{-1} - M(\lambda)^{-1}\right) +
\frac{1}{\lambda + i}\imag(M(i)^{-1})
\end{equation}
holds for a.e. $\lambda\in\dR$ and from \eqref{imm} we obtain
\begin{equation}\label{qqq}
\begin{split}
Q^*Q&=(M(i)^{-1})^*\gamma(i)^*\gamma(i)M(i)^{-1}=(M(i)^{-1})^*\imag (M(i))M(i)^{-1}\\
&=-\imag (M(i)^{-1}).
\end{split}
\end{equation}
Therefore we conclude from \eqref{qqq} and \eqref{settings} that \eqref{zzz}
takes the form
\begin{equation*}
Z(\lambda)=-\frac{1}{\lambda+i}\imag (M(i)^{-1})-\frac{1}{1+\lambda^2}M(i)^{-1}+Q^*(B-\lambda-i0)^{-1}Q 
\end{equation*}
and by inserting \eqref{asdfg} we find $Z(\lambda)=-(1+\lambda^2)^{-1}M(\lambda)^{-1}$. Hence \eqref{scatab1}
turns into the representation \eqref{scatformula} of the scattering
matrix $\{S_{AB}(\lambda)\}$.
\end{proof}

In general it may happen that the operator $S = A \cap B$ is not simple, that is, there
is a nontrivial decomposition of the Hilbert space $\gotH = \wh \gotH \oplus \gotK$ 
such that 
\begin{equation}\label{simsim}
S =  \wh S \oplus H,
\end{equation} 
where $\wh S$ is simple symmetric operator in $\wh\gotH$ and $H$ is a selfadjoint operator in $\gotK$, cf.
Section~\ref{linrel}. Then there exist selfadjoint 
extensions $\widehat A$ and $\widehat B$ of $\widehat S$ in
$\widehat\gotH$ such that 
\begin{equation}\label{abdeco}
A=\widehat A\oplus H\qquad\text{and}\qquad B=\widehat B\oplus H.
\end{equation}
The next result extends the representation of the scattering matrix in Theorem~\ref{scattering} to the case of
a non-simple $S$. 
\begin{thm}\label{III.2}
Let $A$ and $B$ be selfadjoint operators in $\gotH$ such that \eqref{finite}
is satisfied, let $S=A\cap B$ be decomposed as in \eqref{simsim} and let 
$\gP = \{\dC^n,\Gamma_0,\Gamma_1\}$ be a boundary triplet
for $S^*$ such that $A=\ker(\Gamma_0)$ and $B=\ker(\Gamma_1)$; cf. Corollary~\ref{pertuprop}.
Furthermore, let $L^2(\bR,d\gl,\kK_\gl)$ be a spectral representation of the absolutely
continuous part $H^{ac}$ of the selfadjoint operator $H$ in the Hilbert space $\gotK$.

Then $L^2(\bR,d\gl,\kH_\gl
\oplus \kK_\gl)$ is a spectral representation of $A^{ac}$ and the scattering 
matrix $\{S_{AB}(\gl)\}_{\lambda\in\dR}$ is given by
\bed
S_{AB}(\gl) = 
\begin{pmatrix}
S_{\wh A,\wh B}(\gl) & 0\\
0 & I_{\kK_\gl}
\end{pmatrix} 
\in\bigl[
\kH_\gl
\oplus
\kK_\gl\bigr]
\eed
for a.e. $\gl \in \bR$, where $\kH_\gl = \ran(\imag(M(\gl+i0)))$,
$M(\cdot)$ is the Weyl function corresponding to the boundary triplet
$\Pi$ and 
\begin{equation*}
S_{\widehat A\widehat B}(\gl) = I_{\kH_\gl} -
2i\sqrt{\imag(M(\gl))}M(\gl)^{-1}\sqrt{\imag(M(\gl))}\in [\cH_\lambda]
\end{equation*}
is the scattering matrix of the scattering system $\{\widehat
A,\widehat B\}$ from \eqref{abdeco}.
\end{thm}
\begin{proof}
It follows from the decomposition \eqref{abdeco} that the absolutely continuous subspaces
$\gotH^{ac}(A)$ and $\gotH^{ac}(B)$ can be written as the orthogonal
sums 
\bed
\gotH^{ac}(A) = \wh\gotH^{ac}(\wh A) \oplus \gotK^{ac}(H) \quad \text{and}
\quad
\gotH^{ac}(B) = \wh\gotH^{ac}(\wh B) \oplus \gotK^{ac}(H)
\eed
of the absolutely continuous subspaces of $\widehat A$ and $\widehat B$, and the absolutely continuous subspace 
$\gotK^{ac}(H)$ of the selfadjoint operator $H$ in $\gotK$.
Therefore the wave operators of $W_\pm(B,A)$ of the scattering system $\{A,B\}$ 
can be written with the wave operators $W_\pm(\wh B,\wh A)$ of the scattering system $\{\wh A,\wh B\}$
in the form 
\bed
W_\pm(B,A) = W_\pm(\wh B,\wh A) \oplus I_{\gotK^{ac}(H)}.
\eed
This implies the corresponding decomposition of the scattering operator $S_{AB}$ in
$S =  S_{\wh A \wh B} \oplus I_{\gotK^{ac}(H)}$ and hence the scattering matrix 
$\{S_{AB}(\lambda)\}_{\lambda\in\dR}$ of the scattering system $\{A,B\}$ coincides with the orthogonal sum of the 
scattering matrix $\{S_{\widehat A\widehat B}(\lambda)\}_{\lambda\in\dR}$ of the scattering system $\{\widehat A,\widehat B\}$
and the identity operator in the spectral representation $L^2(\bR,d\gl,\kK_\gl)$ of $H^{ac}$.

It is not difficult to see that $\widehat\Pi=\{\dC^n,\widehat\Gamma_0,\widehat\Gamma_1\}$, where
$\widehat\Gamma_0$ and $\widehat \Gamma_1$ denote the restrictions of the boundary mappings $\Gamma_0$ and
$\Gamma_1$ from $S^*=(\widehat S)^*\oplus H$ onto $(\widehat S)^*$, is a boundary triplet for $(\widehat S)^*$
such that $\widehat A=\ker(\widehat\Gamma_0)$ and $\widehat B=\ker(\widehat\Gamma_1)$. Moreover, the Weyl function
corresponding to $\widehat \Pi$ coincides with the Weyl function $M(\cdot)$ corresponding to $\Pi$. Hence, by 
Theorem~\ref{scattering} the
scattering matrix $\{S_{\widehat A\widehat
  B}(\lambda)\}_{\lambda\in\dR}$ is given by \eqref{scatformula}.
\end{proof}

\section{Dissipative and Lax-Phillips scattering systems}\label{sec4}

In this section we consider a scattering systems $\{A,B\}$ 
consisting of a selfadjoint operator $A$ and a 
maximal dissipative operator $B$ in the Hilbert space $\gotH$. 
As above it is assumed that 
\begin{equation}\label{4.1}
(B-\lambda)^{-1}-(A-\lambda)^{-1}\in\cF_n(\gotH),\qquad\lambda\in\rho(A)\cap\rho(B),
\end{equation}
holds for some $n\in\dN$.
Then the closed symmetric operator $S = A \cap B$ is in general not densely defined and its 
deficiency indices are $(n,n)$. By
Corollary~\ref{pertupropcor} there exists a boundary triplet 
$\Pi=\{\dC^n,\Gamma_0,\Gamma_1\}$ for $S^*$ and a
dissipative $n\times n$-matrix $D$ such that $A=\ker(\Gamma_0)$, 
$B=\ker(\Gamma_1-D\Gamma_0)$ and 
\begin{equation*}
(B-\lambda)^{-1}-(A-\lambda)^{-1}=\gamma(\lambda)\bigl(D-M(\lambda)\bigr)^{-1}\gamma(\bar\lambda)^*
\end{equation*}
holds. For our later purposes in Section~\ref{invsec} it is sufficient to investigate
the special case $\ker(\imag (D)) =\{0\}$, 
the general case can be treated in the same way as in
\cite{BMN06b,BMN08a}.

For the investigation of the dissipative scattering system 
$\{A,B\}$ it is useful to construct a so-called minimal
selfadjoint dilation $L$ of the maximal dissipative operator $B$.
For the explicit construction of $L$ we will use the
following lemma which also shows how the constant 
function $\dC_+\ni\lambda\mapsto -i\imag (D)$, $\gl \in \dC_+$, 
can be realized as a Weyl function. 
A detailed proof of Lemma~\ref{glem} can be found in \cite{BMN06b}.
\begin{lem}\label{glem}
Let $T$ be the first order differential operator in the Hilbert space 
$L^2(\dR,\dC^n)$ defined by
\begin{equation*}
(Tg)(x)=-ig^\prime(x),\qquad \dom(T)=\bigl\{g\in W^1_2(\dR,\dC^n):g(0)=0\bigr\}.
\end{equation*}
Then the following holds.
\begin{itemize}
\item [{\rm (i)}]
$T$ is a densely defined closed simple 
symmetric operator with deficiency indices $(n,n)$. 
\item [{\rm (ii)}] The adjoint
operator is 
\begin{equation*}
(T^*g)(x)=-ig^\prime(x),\qquad \dom (T^*)=W^1_2(\dR_-,\dC^n)\oplus W^1_2(\dR_+,\dC^n).
\end{equation*} 
\item [{\rm (iii)}]
The triplet $\Pi_T=\{\dC^n,\Upsilon_0,\Upsilon_1\}$, where 
\begin{equation*}
\begin{split}
  \Upsilon_0\hat g&:=\frac{1}{\sqrt{2}}(-\imag (D))^{-\frac{1}{2}}\,\bigl(g(0+)-g(0-)\bigr),\\
\Upsilon_1\hat g&:=
\frac{i}{\sqrt{2}}(-\imag (D))^{\frac{1}{2}}\,\bigl(g(0+)+g(0-)\bigr),\quad \hat g=\{g,T^*g\},
\end{split}
\end{equation*}
is a boundary triplet for $T^*$ and $T_0=\ker(\Upsilon_0)$ 
is the selfadjoint first order differential
operator in $L^2(\dR,\dC^n)$ defined on $W^1_2(\dR,\dC^n)$. 
\item [{\rm (iv)}]
The Weyl function $\tau(\cdot)$ corresponding to the boundary 
triplet $\Pi_T=\{\dC^n,\Upsilon_0,\Upsilon_1\}$ is given by
\begin{equation*}
\tau(\lambda)=\begin{cases} -i\imag (D), & \lambda\in\dC_+,\\ i\imag (D), & \lambda\in\dC_-.\end{cases}
\end{equation*}
\end{itemize}
\end{lem}

Let $S=A\cap B$ and let $T$ be the first order differential operator 
from Lemma~\ref{glem}. It is clear that
\begin{equation}\label{ortsum}
\begin{pmatrix} S & 0 \\ 0 & T\end{pmatrix}%,\qquad \dom(R)=\dom (S)\oplus\dom (T),
\end{equation}
is a closed symmetric operator in the Hilbert space 
$\gotH\oplus L^2(\dR,\dC^n)$ with deficiency indices $(2n,2n)$ and the
adjoint of \eqref{ortsum} is the orthogonal sum of the relation $S^*$ and 
the operator $T^*$ from Lemma~\ref{glem}.
The next theorem, which is a variant of 
\cite[Theorem 3.2]{BMN06b}, shows how a minimal selfadjoint dilation of the dissipative
operator $B=\ker(\Gamma_1-D\Gamma_0)$ can be constructed. For the particular case of 
Sturm-Liouville operators with dissipative boundary conditions 
this construction goes back to B.S.~Pavlov; cf. \cite{Pav76,Pav96}.
\begin{thm}\label{dilthm}
Let $A$ be a selfadjoint operator and let $B$ be a maximal dissipative operator in $\gotH$ 
such that \eqref{4.1} holds. Let $\Pi=\{\dC^n,\Gamma_0,\Gamma_1\}$ be a boundary triplet
for $S^*$, $S = A\cap B$, and let $D$ be a dissipative
$n\times n$-matrix with $\ker(\imag(D))=\{0\}$ such that  
$A=\ker(\Gamma_0)$ and $B=\ker(\Gamma_1-D\Gamma_0)$; cf. 
Proposition~\ref{pertupropd}. If $\Pi_T=\{\dC^n,\Upsilon_0,\Upsilon_1\}$ is
the boundary triplet of $T^*$ introduced in Lemma \ref{glem}, then the operator
\begin{equation*}
 L\begin{pmatrix} f \\ g\end{pmatrix}=
\begin{pmatrix} f^\prime \\ T^*g\end{pmatrix},\qquad \hat f=\{f,f^\prime\}\in S^*,\,\,\,
\hat g=\{g,T^*g\},
\end{equation*}
with domain
\begin{equation*}
\dom(L)=\left\{\begin{pmatrix} f \\ g\end{pmatrix}\in\dom(S^*)\oplus\dom(T^*):
\begin{matrix}
\Gamma_0\hat f-\Upsilon_0\hat g  =  0\\
(\Gamma_1-\real(D)\Gamma_0)\hat f = -\Upsilon_1\hat g\end{matrix}\right\}
\end{equation*}
is a minimal selfadjoint dilation of the maximal 
dissipative operator $B$, that is, for all $\lambda\in\dC_+$
\begin{equation*}
P_\gotH\bigl(L-\lambda\bigr)^{-1}\upharpoonright_\gotH=(B-\lambda)^{-1}
\end{equation*}
holds and the condition 
$\gotH\oplus L^2(\dR,\dC^n)=\clospa\{(L-\lambda)^{-1}\gotH:\lambda\in\dC\backslash\dR\}$
is satisfied.
\end{thm}
\begin{proof}
Besides the assertion that $L$ is an operator the proof 
of Theorem~\ref{dilthm} is essentially the same as the proof 
of \cite[Theorem 3.2]{BMN06b}. The fact that the restriction 
$L$ of the relation $S^*\oplus T^*$ is an operator
can be seen as follows: Suppose that $\hat f\oplus \hat g\in L$, where
$\{0,f^\prime\}\in S^*$, $\{0,g^\prime\}\in T^*$. Since $T^*$ is an
operator we have $g^\prime=0$ and this implies
$\hat g=0$. Therefore we obtain from the 
boundary conditions in $\dom(L)$ that
\begin{equation*}
\Gamma_0\hat f=\Upsilon_0\hat g=0
\end{equation*}
holds. Hence $\hat f=\{f,f^\prime\}$ belongs to 
$A=\ker(\Gamma_0)$ which is an operator. Therefore $f^\prime=0$ and $L$
is an operator.
\end{proof}

Let $L$ be the minimal selfadjoint dilation of the 
maximal dissipative operator $B$ from Theorem~\ref{dilthm}
and define a selfadjoint operator $K$ in $\gotH\oplus L^2(\dR,\dC^n)$
by
\begin{equation}\label{l}
K:=\begin{pmatrix} A & 0 \\ 0 & T_0\end{pmatrix},
\end{equation}
where $T_0=\ker(\Upsilon_0)$ is the selfadjoint first order 
differential operator from Lemma~\ref{glem}. In the following theorem we 
consider the scattering system $\{K,L\}$ in the Hilbert 
space $\gotH\oplus L^2(\dR,\dC^n)$. The operator $R := K \cap L$
is symmetric and may have a nontrivial selfadjoint part $H$ which acts
in the Hilbert space 
\begin{equation*}
 \bigl(\gotH\oplus L^2(\dR,\dC^n)\bigr)\ominus\clospa\bigl\{\ker(R^*-\lambda):\lambda\in\dC\backslash\dR\bigr\}.
\end{equation*}
Hence the operators $K$ and $L$ admit the decompositions
\bed
K = \wh K \oplus H 
\quad \mbox{and} \quad 
L = \wh L \oplus H,
\eed
with selfadjoint operators $\wh K$ and $\wh L$ in $\clospa\{\ker(R^*-\lambda):\lambda\in\dC\backslash\dR\}$ 
and we have $R = \wh R \oplus H$, where $\wh R = \wh K \cap \wh L$. In
particular, $\wh K$ and $\wh L$ are both selfadjoint extensions of the closed simple symmetric operator $\wh R$.
We remark that the symmetric operator $R$ is an $n$-dimensional extension of the orthogonal sum in \eqref{ortsum};
this follows easily from the next theorem.
In the following we assume that $L^2(\bR,d\gl,\kK_\gl)$ is a spectral representation of the absolutely continuous 
part $H^{ac}$ of $H$.
\begin{thm}\label{couplescat}
Let $A$ be a selfadjoint operator and let $B$ be a maximal dissipative operator in $\gotH$ 
such that \eqref{4.1} holds. Let $\Pi=\{\dC^n,\Gamma_0,\Gamma_1\}$ be a boundary triplet
for $S^*$, $S = A\cap B$, and let $D$ be a dissipative $n\times n$-matrix with $\ker(\imag(D))=\{0\}$
such that  
$A=\ker(\Gamma_0)$ and $B=\ker(\Gamma_1-D\Gamma_0)$; cf. Proposition~\ref{pertupropd}. 
If $L$ is the minimal self-adjoint dilation of $B$
in Theorem \ref{dilthm} and $K$ is given by \eqref{l}, then
\begin{equation}\label{klresdiff}
(K-\lambda)-(L-\lambda)^{-1}\in\cF_n,\qquad\lambda\in\dC\backslash\dR.
\end{equation}
Moreover, if $L^2(\bR,d\gl,\kK_\gl)$ is a spectral representation of
$H^{ac}$, where $H$ is the maximal self-adjoint part of $R = K\cap
L$, then $L^2(\bR,d\gl,\dC^n\oplus\kK_\gl)$ is a spectral representation of $K$ 
and the scattering matrix $\{S_{KL}(\lambda)\}_{\lambda\in\dR}$ of 
the scattering system $\{K,L\}$ admits the representation
\begin{equation*}
S_{KL}(\lambda)=
\begin{pmatrix} 
S_{\wh K\wh L}(\lambda) & 0 \\ 
0 & I_{\kK_\gl}
\end{pmatrix}\in\bigl[\dC^n \oplus \kK_\gl\bigr]
\end{equation*}
for a.e. $\lambda\in\dR$, where 
\begin{equation*}
S_{\wh K\wh L}(\lambda)=I_{\dC^n}-2i\sqrt{\imag
  (M(\lambda)-D)}\,\bigl(M(\lambda)-D\bigr)^{-1}
\sqrt{\imag (M(\lambda)-D)}
\end{equation*}
is the scattering matrix of the scattering system $\{\widehat K,\widehat L\}$,
$M(\cdot)$ is the Weyl function of the boundary triplet $\Pi$ and 
$M(\lambda)=M(\lambda+i0)$. 
\end{thm}
\begin{proof}
We are going to apply Theorem \ref{III.2} to the scattering system
$\{K,L\}$. 
For this we consider the 
symmetric operator $R=K\cap L$ and note that the operator $K$ is given by
$\ker(\Gamma_0)\oplus\ker(\Upsilon_0)$.
Hence the boundary condition 
$\Gamma_0\hat f-\Upsilon_0\hat g=0$ in $\dom(L)$ is automatically fulfilled and this implies
that the intersection $R=K\cap L$
given by 
\begin{equation*}
\begin{split}
R\begin{pmatrix} f \\ g\end{pmatrix}&=K
\begin{pmatrix} f \\ g\end{pmatrix}=L\begin{pmatrix} f \\ g\end{pmatrix}\\
\dom(R)&=\left\{\begin{pmatrix} f \\ g\end{pmatrix}
\in\dom K:(\Gamma_1-\real(D)\Gamma_0)\hat f = -\Upsilon_1\hat g\right\},
\end{split}
\end{equation*}
where $\hat f=\{f,Af\}$ and $\hat g=\{g,T_0g\}$. 
It is not difficult to verify that the adjoint operator $R^*$  has the form
\begin{equation*}
R^*=\left\{\hat f\oplus\hat g\in S^*\oplus T^*:\Gamma_0\hat f=\Upsilon_0\hat g\right\}
\end{equation*}
and $\widetilde\Pi=\{\dC^n,\widetilde\Gamma_0,\widetilde\Gamma_1\}$,
where
\begin{equation*}
\widetilde\Gamma_0(\hat f\oplus\hat g)=\Gamma_0\hat f\quad\text{and}\quad\widetilde\Gamma_1(\hat f\oplus\hat g)=
(\Gamma_1-\real(D)\Gamma_0)\hat f+\Upsilon_1\hat g,
\end{equation*}
is a boundary triplet for $R^*$. Observe that $K=\ker(\widetilde\Gamma_0)$ and $L=\ker(\widetilde\Gamma_1)$. 
This also implies that the difference of the resolvents of $K$ and $L$ in \eqref{klresdiff} is a rank $n$ operator;
cf. Corollary~\ref{pertuprop}.

Let us compute the Weyl function $\widetilde M$ corresponding to the boundary triplet $\widetilde\Pi$.
For $\lambda\in\dC_+$ and $\hat f\oplus\hat g\in\widehat\cN_{\lambda,R^*}$ we have 
$\hat f\in\widehat\cN_{\lambda,S^*}$, 
$\hat g\in\widehat\cN_{\lambda,T^*}$ and $\Gamma_0\hat f=\Upsilon_0\hat g$. 
Hence the definition of the Weyl function and Lemma~\ref{glem} imply
\begin{equation*}
\begin{split}
\widetilde M(\lambda)\widetilde\Gamma_0(\hat f\oplus\hat g)&=
\widetilde\Gamma_1(\hat f\oplus \hat g)=
\Gamma_1\hat f-\real(D)\Gamma_0\hat f+\Upsilon_1\hat g\\
&=M(\lambda)\Gamma_0\hat f-\real(D)\Gamma_0\hat f- i\imag (D)\Upsilon_0\hat g\\
&=(M(\lambda)-D)\Gamma_0\hat f=(M(\lambda)-D)\widetilde\Gamma_0(\hat f\oplus\hat g)
\end{split}
\end{equation*}
and therefore $\widetilde M(\lambda)=M(\lambda)-D$ 
for $\lambda\in\dC_+$. As $D$ is a dissipative matrix and $\ker(\imag (D))=\{0\}$
by assumption it follows that 
$$\imag(\widetilde M(\lambda+i0))=\imag(M(\lambda+i0))-\imag(D)$$ 
is uniformly positive and hence
$\ran(\imag (\widetilde M(\lambda+i0)))=\dC^n$. Now Theorem~\ref{III.2} 
applied to the boundary triplet $\widetilde\Pi$ and the corresponding Weyl function $\widetilde M$ 
yields the statement of Theorem~\ref{couplescat}.
\end{proof}

For our later purposes it is useful to express the 
scattering matrix $\{S_{KL}(\lambda)\}_{\lambda\in\dR}$ in Theorem~\ref{couplescat}
in a slightly different form. The following proposition extends \cite[Theorem 3.6]{BMN06b} to the case where $S=A\cap B$
is not necessarily densely defined. The proof is almost the same and will not be repeated.
\begin{prop}\label{sklprop}
Let the assumptions of Theorem~\ref{couplescat} be satisfied, assume, in addition, that $S=A\cap B$ is simple and 
let $L^2(\bR,d\gl,\kH_\gl)$, $\kH_\gl = \ran(\imag(M(\gl)))$, $M(\gl)
:= M(\gl+i0)$, be a spectral representation of $A^{ac}$. 

Then
$L^2(\bR,d\gl,\kH_\gl \oplus \dC^n)$ is a spectral representation of
$K^{ac} = A^{ac} \oplus T_0$ such that the scattering matrix $\{S_{KL}(\lambda)\}_{\gl \in \bR}$
of the scattering system $\{K,L\}$ can be expressed by 
\begin{equation*}
S_{KL}(\lambda)=
\begin{pmatrix} 
I_{\cH_\lambda} & 0 \\ 0 & I_{\dC^n} 
\end{pmatrix}+
2i\begin{pmatrix} 
S_{11}(\lambda) & S_{12}(\lambda) \\ 
S_{21}(\lambda) & S_{22}(\lambda) 
\end{pmatrix}
\in\bigl[
\cH_\lambda
\oplus \dC^n\bigr]
\end{equation*}
for a.e. $\lambda\in\dR$, where
\begin{equation*}
\begin{split}
S_{11}(\lambda)&=\sqrt{\imag(M(\lambda))}\,\bigl(D-M(\lambda)\bigr)^{-1}\sqrt{\imag (M(\lambda))},\\
S_{12}(\lambda)&=\sqrt{\imag(M(\lambda))}\,\bigl(D-M(\lambda)\bigr)^{-1}\sqrt{-\imag (D)},\\
S_{21}(\lambda)&=\sqrt{-\imag(D)}\,\bigl(D-M(\lambda)\bigr)^{-1}\sqrt{-\imag (M(\lambda))},\\
S_{22}(\lambda)&=\sqrt{-\imag(D)}\,\bigl(D-M(\lambda)\bigr)^{-1}\sqrt{-\imag (D)}.
\end{split}
\end{equation*}
\end{prop}
\begin{rem}
If $S=A\cap B$ is simple we find by combining Theorem  \ref{couplescat} with Proposition \ref{sklprop} that 
$\dim(\kK_\gl) = \dim(\kH_\gl)$ holds for
a.e. $\gl \in \dR$, i.e., the spectral multiplicity of $H^{ac}$,
where $H$ is the maximal self-adjoint part of $R = K \cap L$ is equal to
the spectral multiplicity of $A^{ac}$.
\end{rem}

In the following we are going to interpret the right lower corner $I+2iS_{22}$
of the scattering matrix $\{S_{KL}(\lambda)\}_{\lambda\in\dR}$ 
in Proposition~\ref{sklprop} as the scattering matrix corresponding
to a Lax-Phillips scattering system; see, e.g., \cite{BW83,LP67} for further details.
For this purpose we decompose 
$L^2(\dR,\dC^n)$ into the orthogonal sum of the subspaces
\bed
\kD_-:=L^2(\dR_-,\dC^n)\quad\text{and}\quad\kD_+:=L^2(\dR_+,\dC^n),
\eed
and denote the natural embeddings of $\cD_\pm$ into $\gotH\oplus L^2(\dR,\dC^n)$ by $J_\pm$.
The subspaces $\cD_+$ and $\cD_-$ are called
{\it outgoing} and {\it incoming subspaces}, respectively, for the selfadjoint dilation
$L$ in $\gotH\oplus L^2(\dR,\dC^n)$, i.e.
\begin{equation*}
e^{-itL} \cD_\pm \subseteq\cD_\pm,\,\,\,\,t\in\dR_\pm,  \quad\text{and}\quad
\bigcap_{t\in \bR}e^{-itL}\kD_\pm = \{0\}
\end{equation*}
hold. The system $\{L,\cD_-,\cD_+\}$ is called {\it Lax-Phillips scattering system}
and the corresponding {\it Lax-Phillips wave operators} are defined by
\bed
\gO_\pm := \slim_{t\to\pm\infty}e^{itL}J_\pm
e^{-itT_0}:L^2(\dR,\dC^n)\rightarrow\gotH\oplus L^2(\dR,\dC^n);
\eed
cf. \cite{BW83}. Since $\slim_{t\to\pm\infty}J_\mp e^{-itT_0} = 0$
the restrictions of the wave operators $W_\pm(L,K)$
of the scattering system $\{K,L\}$ onto $L^2(\dR,\dC^n)$
coincide with the Lax-Phillips wave operators $\Omega_\pm$ and hence 
the {\it Lax-Phillips scattering operator} $S^{LP} := \gO^*_+\gO_-$
is given by $S^{LP}= P_{L^2}S_{KL}\,\iota_{L^2}$,
where $S_{KL}$ is the scattering operator
of the scattering system $\{K,L\}$, $P_{L^2}$ is 
the orthogonal projection from $\gotH\oplus L^2(\dR,\dC^n)$ onto  
$L^2(\dR,\dC^n)$ and $\iota_{L^2}$ denotes the canonical embedding. 
Hence the Lax-Phillips scattering operator $S^{LP}$
is a contraction in $L^2(\dR,\dC^n)$ and commutes with 
the selfadjoint differential operator $T_0$.
Therefore $S^{LP}$ is unitarily equivalent to a multiplication operator
induced by a family $\{S^{LP}(\lambda)\}_{\lambda\in\dR}$ of 
contractive operators in $L^2(\dR,\dC^n)$; this family
is called the {\it Lax-Phillips scattering matrix}.

The above considerations together with Proposition~\ref{sklprop} immediately
imply the following corollary on the representation of 
the Lax-Phillips scattering matrix; cf. \cite[Corollary 3.10]{BMN06b}.
\begin{cor}\label{lax1}
Let the assumptions of Proposition~\ref{sklprop} be satisfied. 
If $\{L,\cD_-,\cD_+\}$ is the Lax-Phillips scattering
system from above, then the Lax-Phillips scattering matrix
$\{S^{LP}(\gl)\}_{\lambda\in\dR}$ admits the representation
\bed
S^{LP}(\lambda)= I_{\dC^n}+2i\sqrt{-\imag(D)}\bigl(D-M(\lambda)\bigr)^{-1}
\sqrt{-\imag(D)}
\eed
for a.e. $\lambda\in\bR$, where $M(\lambda)=M(\lambda+i0)$.
\end{cor}
We mention that Corollary~\ref{lax1} also implies a 
well-known result of Adamyan and Arov in \cite{AA65a,AA65b,AA66a,AA66b}
on the relation of the Lax-Phillips scattering matrix with the 
characteristic function of the maximal dissipative operator $B$;
see \cite{BMN06b} for further details.

\section{An inverse scattering problem}\label{invsec}

Let $W:\dC_+\rightarrow [\dC^n]$ be a contractive analytic 
matrix function defined on the upper half-plane $\dC_+$.
Then the limit 
\begin{equation*}
W(\gl) = W(\gl+i0)=\lim_{y\to+0}W(\gl + iy)
\end{equation*}
exists for a.e. $\gl \in \bR$. In the following 
theorem we show that under some mild additional conditions the limit of the function $W$ 
can be regarded as the scattering matrix of a Lax-Phillips 
scattering system $\{L,\cD_-,\cD_+\}$, where $L$ is the minimal
selfadjoint dilation of some maximal dissipative operator 
in a Hilbert space $\gotH$ as in the previous section.
\begin{thm}\label{invthm}
Let $W: \dC_+ \rightarrow [\dC^n]$ be a contractive analytic
function such that the conditions 
\begin{equation}\label{5.2}
\ker(I - W(\eta)^*W(\eta)) = \{0\}, \qquad  \eta
\in \dC_+,
\end{equation}
and
\begin{equation}\label{5.3}
\lim_{y\to+\infty}\frac{1}{y}(I - W(iy))^{-1} = 0
\end{equation}
are satisfied. Then the following holds:

\begin{enumerate}

\item[\rm (i)]
There exists a separable Hilbert space $\gotH$, a (in general
nondensely defined) simple symmetric operator $S$ with deficiency indices $(n,n)$ in $\gotH$, 
a boundary triplet $\gP = \{\dC^n,\Gamma_0,\Gamma_1\}$ for $S^*$ with Weyl function $M(\cdot)$ 
and a dissipative matrix $D\in[\dC^n]$ with $\ker(\imag(D))=\{0\}$ such that $W(\cdot)$ admits the representation
\begin{equation}\label{repw}
W(\mu)= I + 2i\sqrt{-\imag(D)}\bigl(D - M(\mu)\bigr)^{-1}\sqrt{-\imag(D)}
\end{equation}
for all $\mu\in\dC_+$ and a.e. $\mu\in\dR$, where $W(\mu)=W(\mu+i0)$ and $M(\mu)=M(\mu+i0)$.

\item[\rm (ii)]  
The function $\dR\ni\mu\mapsto W(\mu)$ is
the Lax-Phillips scattering matrix of the Lax-Phillips scattering
system $\{L,\cD_-,\cD_+\}$, where 
$L$ is the minimal selfadjoint dilation of the maximal dissipative extension
$B = \ker(\gG_1 - D\gG_0)$ of $S$ in 
Theorem~\ref{dilthm} and the incoming and outgoing subspaces are $\cD_\pm=L^2(\dR_\pm,\dC^n)$.

\item[\rm (iii)]
If, in addition, the condition
\begin{equation}\label{5.5}
\lim_{y\to +\infty}y^{1/2}\left\|\sqrt{I -
       W(iy)^*W(iy)}\;(I - W(iy))^{-1}h\right\| = \infty
\end{equation}
holds for all $h\in\kH$, $h\not=0$, then the symmetric operator $S$ is densely defined.

\item[\rm (iv)]
If $W(\cdot)$ is an inner function, i.e. $W(\mu+ i0)$ is unitary for
a.e. $\mu \in \bR$, then the spectrum of the selfadjoint operator 
$A = \ker(\gG_0)$ is purely singular and the absolutely continuous part $A^{ac}$ of $A$ is trivial.
\end{enumerate}
\end{thm}
\begin{proof}
(i) 
Observe that condition \eqref{5.2} together with 
$\Vert W(\eta)\Vert\leq 1$ implies that
$\ker(I-W(\eta))=\{0\}$ holds for all $\eta\in\dC_+$. Indeed, 
for $x\in\ker(I-W(\eta))$ we have
$\Vert W(\eta) x\Vert=\Vert x\Vert$ and hence 
\begin{equation*}
\bigl((I-W(\eta)^*W(\eta))x,x\bigr)=0
\end{equation*}
which yields $x = 0$. 
We define a function $M:\dC_+\rightarrow [\dC^n]$ by
\begin{equation}\label{mfct}
\dC_+\ni\eta\mapsto M(\eta):= i(I + W(\eta))(I - W(\eta))^{-1}
\end{equation}
and we extend $M$ to the lower half-plane by $\dC_-\ni\eta\mapsto M(\eta):=M(\bar\eta)^*$.
Then $M$ is analytic and  a straightforward computation shows
\begin{equation}\label{5.5a}
\imag(M(\eta)) = 
(I - W(\eta)^*)^{-1}(I - W(\eta)^*W(\eta))(I - W(\eta))^{-1} \ge 0.
\end{equation}
for $\eta\in\dC_+$.
Hence $M$ is a Nevanlinna function and condition \eqref{5.2} implies 
$\ker(\imag(M(\eta))) = \{0\}$.
From condition \eqref{5.3} we obtain 
\begin{equation*}
\slim_{y\to+\infty}\frac{1}{y}M(iy) = 
\slim_{y\to+\infty}\frac{i}{y} (I + W(iy))(I - W(iy))^{-1}= 0.
\end{equation*}
By Theorem~\ref{nevweyl} there exists a separable Hilbert 
space $\gotH$, a (in general nondensely defined) 
simple symmetric operator $S$ with deficiency indices $(n,n)$ 
and a boundary triplet $\Pi=\{\dC^n,\Gamma_0,\Gamma_1\}$ 
for $S^*$ such that $M$ is the corresponding
Weyl function. For $\eta\in\dC_+$ we have
\begin{equation}\label{5.7}
W(\eta)  = I - 2i(M(\eta) + i)^{-1}, \quad \eta \in \dC_+. 
\end{equation}
Setting $D:=-iI$ we have $\sqrt{-\imag(D)}=I$ and hence the representation \eqref{repw} follows from \eqref{5.7}.
\vskip 0.1cm \noindent
(ii) From  Corollary~\ref{lax1} one immediately gets that  
$W$ can be regarded as the Lax-Phillips scattering matrix 
of the Lax-Phillips scattering system $\{L,\cD_-,\cD_+\}$.
\vskip 0.1cm \noindent
(iii)
Making use of \eqref{5.5a} one easily verifies that the condition
\eqref{5.5} yields 
\begin{equation*}
\lim_{y\to\infty}y\,\imag(M(iy)h,h) = \infty,\qquad h \in \kH \setminus
\{0\}.
\end{equation*} 
Hence the operator $S$ is densely defined by Theorem \ref{nevweyl}.
\vskip 0.1cm \noindent
(iv) We consider the analytic function $w(\eta) := \det(I
- W(\eta))$, $\eta \in \bC_+$. Since the limit $W(\gl + i0) :=
\lim_{y\to+0}W(\gl + iy)$ exists for a.e. $\gl \in \dR$ the 
limit $w(\gl + i0) := \lim_{y\to+0}w(\gl + iy)$ exist for a.e. $\gl
\in \dR$, too. If the Lebesgue measure of the set $\{\gl \in \dR: w(\gl +
i0) = 0\}$ is different from zero, then $w(\eta) \equiv 0$ for all $\eta \in
\dC_+$ by the Lusin-Privalov theorem \cite[Section III]{Koo80}
  but this is impossible by assumption \eqref{5.2}, cf. proof of (i).
Hence, the set $\{\gl \in \dR: w(\gl + i0) = 0\}$ has Lebesgue
measure zero.  Therefore, the operator $(I - W(\gl + i0))^{-1}$ exist
for a.e. $\gl \in \dR$. Using \eqref{5.5a} we find that
$\lim_{y\to+0}\imag(M(\gl + iy)) = 0$ for a.e. $\gl \in \dR$. By
\cite[Theorem 4.3(iii)]{BMN02} we get that the selfadjoint operator $A =\ker(\gG_0)$ has no absolutely continuous spectrum, i.e.,
the absolutely continuous part of $A$ is trivial. 
\end{proof}

We remark, that
the representation \eqref{repw} can also be obtained from \cite[Proposition 7.5]{DM95}. In fact, in the
special case considered here some parts of
the proof of Theorem~\ref{invthm}~(i) coincide with the proof of \cite[Proposition 7.5]{DM95}.

The Lax-Phillips scattering system and the selfadjoint dilation $L$ in Theorem~\ref{invthm} can be made more explicit. Let $W:\dC_+\rightarrow [\dC^n]$ be
as in the assumptions of Theorem~\ref{invthm} and define the function
$M$ by
\begin{equation*}
M(\eta)=i(I+W(\eta))(I-W(\eta))^{-1},\qquad \eta\in\dC_+
\end{equation*}
and $M(\eta)=M(\bar\eta)^*$ as in \eqref{mfct}. Then $M$ is $[\dC^n]$-valued Nevanlinna function and hence 
$M$ admits an integral representation of the form
\begin{equation}\label{intrep}
 M(\eta)=\alpha + 
\int_{\dR}\left(\frac{1}{t-\eta}-\frac{t}{1+t^2}\right)d\Sigma(t),\qquad\eta\in\dC\backslash\dR,
\end{equation}
where $\alpha$ is a symmetric matrix and $t\mapsto \Sigma(t)$ is a $[\dC^n]$-valued nondecreasing symmetric matrix function on $\dR$
such that $\int(1+t^2)^{-1}d\Sigma(t)\in[\dC^n]$. We note that due to condition (ii) in 
Theorem~\ref{invthm} the linear term in the integral representation \eqref{intrep} is absent.
Let $L^2_\Sigma(\dR,\dC^n)$ be the Hilbert space of $\dC^n$-valued
functions as in \cite{Berez68,GG91,Mal-Mal03}. 
%with the property 
%
%
%\begin{equation*}
%\Vert g\Vert_{L^2_\Sigma}:=
%\left(\int_\dR \bigl(d\Sigma(t)g(t),g(t)\bigr)\right)^{\frac{1}{2}}<\infty
%\end{equation*}
%
%
%and denote the corresponding scalar product by $(\cdot,\cdot)_{L^2_\Sigma}$, for details on the
%space $L^2_{\Sigma}(\dR,\dC^n)$; see, e.g., \cite{Berez68,GG91,Mal-Mal03}. 
It was shown in \cite{Mal-Mal03} that the mapping 
$$f\mapsto \int_\dR d\Sigma(t)f(t)$$
defined originally on the space $C_0(\dR,\dC^n)$ of (strongly) continuous functions with compact support
admits a continuous extension to an operator from $L^2_\Sigma(\dR,\dC^n)$ into $\dC^n$.
According to \cite{DM95,Mal-Mal03} 
the adjoint of the (in general nondensely defined) closed
symmetric operator
\begin{equation*}
\begin{split}
(Sf)(t)&=tf(t),\\
\dom(S)&=\left\{f\in L^2_\Sigma(\dR,\dC^n):tf(t)\in L^2_\Sigma(\dR,\dC^n),\int_\dR d\Sigma(t)f(t)=0\right\},
\end{split}
\end{equation*}
is given by the linear relation 
\begin{equation*}
S^*=\left\{\left\{f(t)+\frac{t}{1+t^2}h,tf(t)-\frac{1}{1+t^2}h\right\}:f(t),tf(t)\in L^2_\Sigma(\dR,\dC^n),\,h\in\dC^n\right\},
\end{equation*}
and that $\{\dC^n,\Gamma_0,\Gamma_1\}$, where
\begin{equation*}
\Gamma_0\hat f:=h\qquad\text{and}\qquad \Gamma_1\hat f:=\alpha h+\int_\dR d\Sigma(t) f(t),
\end{equation*}
$\hat f=\{f(t)+t(1+t^2)^{-1}h,tf(t)-(1+t^2)^{-1}h\}\in S^*$,
is a boundary triplet for $S^*$ with corresponding Weyl function $M(\cdot)$. Note that here $A_0=\ker(\Gamma_0)$ is the 
usual maximal multiplication operator in $L^2_\Sigma(\dR,\dC^n)$. 
\begin{cor}\label{invcor}
Let $W: \bC_+ \rightarrow [\bC^n]$ be a contractive analytic
function which satisfies the conditions \eqref{5.2} and \eqref{5.3} in Theorem~\ref{invthm}. Then there exists a symmetric matrix $\alpha\in[\dC^n]$
and a $[\dC^n]$-valued nondecreasing symmetric matrix function $\Sigma(\cdot)$ on $\dR$ such
that
\begin{equation*}
W(\mu)=I-2i\left(\alpha+i+\int_{\dR}\left(\frac{1}{t-\mu}-
\frac{t}{1+t^2}\right)d\Sigma(t)\right)^{-1}
\end{equation*}
holds for all $\mu\in\dC_+$ and $\int(1+t^2)^{-1}d\Sigma(t)\in[\dC^n]$. The function $W(\cdot)$ coincides with the Lax-Phillips scattering matrix
of the system $\{L,\cD_-,\cD_+\}$, where $\cD_\pm=L^2(\dR_\pm,\dC^n)$
and 
\begin{equation*}
\begin{split}
L\begin{pmatrix} f(t)+\frac{t}{1+t^2}h\\ g\end{pmatrix} &=\begin{pmatrix} t f(t)-\frac{1}{1+t^2} h\\ ig^\prime\end{pmatrix},\qquad 
\begin{matrix} f(t),t f(t)\in L^2_\Sigma(\dR,\dC^n),\,h\in\dC^n,\\ g\in W^1_2(\dR_-,\dC^n)\oplus W^1_2(\dR_+,\dC^n),\end{matrix}\\
\dom (L)=&\left\{\begin{pmatrix} f(t)+\frac{t}{1+t^2}h\\ g\end{pmatrix}:\begin{matrix} \frac{1}{\sqrt{2}}(g(0+)-g(0-))=h \\ 
\frac{i}{\sqrt{2}} (g(0+)+g(0-))=\alpha h+\int d\Sigma(t)f(t)\end{matrix}\right\}
\end{split}
\end{equation*}
is the minimal selfadjoint dilation in 
$L^2_\Sigma(\dR,\dC^n)\oplus L^2(\dR,\dC^n)$ of the maximal dissipative multiplication 
operator $B=\ker(\Gamma_1+i\Gamma_0)$ in $L^2_\Sigma(\dR,\dC^n)$.
\end{cor}

\begin{appendix}\label{appen}

\section{Spectral representations and scattering matrix}\label{app}

Let $A$ be a selfadjoint operator in the separable Hilbert space $\gotH$ 
and let $E(\cdot)$ be the corresponding spectral measure defined on the $\gs$-algebra $\kB(\bR)$ of Borel subsets of $\bR$. 
The absolutely continuous and singular part of the measure $E(\cdot)$ is denoted by $E^{ac}(\cdot)$ and $E^s(\cdot)$, respectively.
If $C\in[\kH,\gotH]$ is a Hilbert-Schmidt operator, then by
\cite[Lemma I.11]{BW83} 
\begin{equation*}
\gS(\gd) := C^*E(\gd)C,\qquad \gd \in \kB(\bR),
\end{equation*}
is a trace class valued measure on $\kB(\bR)$ of finite variation.
This measure admits a unique decomposition 
\bed
\gS(\cdot) =  \gS^{s}(\cdot) + \gS^{ac}(\cdot)
\eed
into a singular measure $\gS^{s}(\cdot) = C^*E^s(\cdot)C$ and an
absolutely continuous measure $\gS^{ac}(\cdot) = C^*E^{ac}(\cdot)C$.
According to \cite[Proposition I.13]{BW83} the trace class valued function $\gl\mapsto\gS(\gl) :=
C^*E((-\infty,\gl))C$ admits a derivative $K(\gl) :=
\frac{d}{d\gl}\gS(\gl) \ge 0$ in the trace class norm
for a.e. $\gl \in \bR$ with respect  to the Lebesgue measure $d\gl$ and 
\bed
\gS^{ac}(\gd) = \int_\gd K(\gl) d\gl, \quad \gd \in \kB(\bR)
\eed
holds.
By $\kH_\gl := \overline{\ran(K(\gl))} \subseteq \kH$ we define a
measurable family of subspaces in $\kH$. Let $P(\gl)$ be the
orthogonal projection from $\kH$ onto $\kH_\gl$ and define a
measurable family of projections by
\bed
(Pf)(\gl) := P(\gl)f(\gl), \qquad f \in L^2(\bR,d\gl,\gotH).
\eed
Then $P$ is an orthogonal projection in $L^2(\bR,d\gl,\kH)$ and we denote the range of $P$ by 
$L^2(\bR,d\gl,\kH_\gl)$. In the following we regard
$L^2(\bR,d\gl,\kH_\gl)$ as the direct integral of the measurable
family of subspaces $\{\kH_\gl\}_{\gl \in \bR}$.
\bl\la{A.I}
Let $A$, $E$, $C$ and $K(\lambda)$ be as above and assume that the absolutely continuous subspace $\gotH^{ac}(A)$ satisfies
the condition 
$$\gotH^{ac}(A) =\clospa\{E^{ac}(\gd)\ran(C):\gd \in \kB(\bR)\}.$$
Then the linear extension of the mapping
\begin{equation}\label{specmap}
E^{ac}(\gd)Cf \mapsto \chi_{\gd}(\gl)\sqrt{K(\gl)}f\,\,\,\text{for a.e.}\,\,\lambda\in\dR,\qquad f\in\cH,
\end{equation}
onto the dense subspace $\spa\{E^{ac}(\gd)\ran(C):\gd \in \kB(\bR)\}$ of $\gotH^{ac}(A)$
admits a unique continuation to an isometric isomorphism from $\Phi:\gotH^{ac}(E)\rightarrow L^2(\bR,d\gl,\kH_\gl)$ 
such that
\begin{equation}\label{A.1}
(\Phi E^{ac}(\gd) g)(\gl)  = \chi_\gd(\gl)(\Phi g)(\gl), \qquad g \in \gotH^{ac}(A),
\end{equation}
holds for any $\gd \in \kB(\bR)$. 
\el
\begin{proof}
For $f\in\cH$ and $\delta\in\cB(\dR)$ we have
\bed
\Vert \chi_{\delta}(\cdot)\sqrt{K(\cdot)}f\Vert^2 =
\int_\gd \|\sqrt{K(\gl)}f\|^2_\kH d\gl = \|E^{ac}(\gd)Cf\|^2_\gotH 
\eed
and hence the extension of the mapping \eqref{specmap} onto the subspace 
$\spa\{E^{ac}(\gd)\ran(C):\gd \in \kB(\bR)\}$ is an isometry into $L^2(\bR,d\gl,\kH_\gl)$.
Then the unique extension $\Phi:\gotH^{ac}(A)\rightarrow L^2(\bR,d\gl,\kH_\gl)$ is isometric and
it remains to show that $\Phi$ is onto. Suppose that there exists $h
\in L^2(\bR,d\gl,\kH_\gl)$ such that 
\bed
0 = (\Phi E^{ac}(\gd)Cf,h) = \int_\gd\;(\sqrt{K(\gl)}f,h(\gl))_\kH\,d\gl
\eed
holds for all $\delta\in\cB(\dR)$ and $f\in\cH$. This implies 
$(\sqrt{K(\gl)}f,h(\gl))_\gotH =0$ for a.e. $\gl \in \bR$ and hence $h(\gl) \perp
\kH_\gl$ for a.e. $\gl \in \bR$, thus $h(\gl) = 0$ for a.e. $\gl \in
\bR$. Hence $\Phi$ is surjective. The relation \eqref{A.1}
for $\Phi$ follows from \eqref{specmap}. 
\end{proof}

From \eqref{A.1} we immediately get that
the maximal multiplication operator $\kQ$ in $L^2(\bR,d\gl,\kH_\gl)$, 
 \bead
(\kQ f)(\gl) & := & \gl f(\gl),\\
f \in \dom(\kQ) & := & \{f \in
  L^2(\bR,d\gl,\kH_\gl): \gl f(\gl) \in L^2(\bR,d\gl,\kH_\gl)\}.
\eead
satisfies 
 $\kQ\Phi = \Phi A^{ac}$ and $\varphi(\kQ)\Phi = \Phi\varphi(A^{ac})$ 
for any bounded Borel
measurable function $\varphi(\cdot): \bR \longrightarrow \bR$.
In other words, the direct integral $L^2(\bR,d\gl,\kH_\gl)$ performs
a spectral representation of the absolutely continuous part $A^{ac}$ of the selfadjoint operator $A$.

Suppose now that $B$ is also a selfadjoint operator in $\gotH$ and assume that 
the difference of the resolvents 
\begin{equation*}
(B-i)^{-1}-(A-i)^{-1}
\end{equation*}
is a trace class operator. Then the wave operators
\bed
W_\pm(B,A) := s-\lim_{t\to\pm\infty}e^{itB}e^{-itA}P^{ac}(A)
\eed
exists and are complete; cf., e.g., \cite[Theorem I.1]{BW83}. The scattering 
operator $S_{AB} := W_+(B,A)^*W_-(B,A)$ regarded as an operator in $\gotH^{ac}(A)$
is unitary and commutes with $A$. Therefore there is a
measurable family $\{S_{AB}(\gl)\}_{\gl \in \bR}$ of unitary operators
$ S_{AB}(\gl)\in [\kH_\gl]$
such that $S_{AB}$ is unitarily equivalent to the multiplication operator
$\kS_{AB}$ induced by $\{S_{AB}(\gl)\}_{\gl \in \bR}$ in
$L^2(\bR,d\gl,\kH_\gl)$, that is, $\kS_{AB} = \Phi S_{AB} \Phi^{-1}$.
The measurable family $\{S_{AB}(\gl)\}_{\gl \in \bR}$ is called the
scattering matrix of the scattering system $\{A,B\}$. 

The following theorem on the representation of the scattering matrix is an important ingredient in the proof
of Theorem~\ref{scattering}. A detailed proof of Theorem~\ref{A.II} will appear in a forthcoming paper.
\bt\la{A.II}
Let $A$ and $B$ be selfadjoint operators in the separable Hilbert
space $\gotH$ and suppose that the resolvent difference 
admits the factorization
\bed
 (B-i)^{-1}-(A-i)^{-1}=\phi(A)CGC^*=QC^*,
\eed
where $C\in[\kH,\gotH]$ is a Hilbert-Schmidt
operator, $G\in[\kH]$,  
$\phi(\cdot): \bR \rightarrow \bR$ is a bounded continuous
function and $Q = \phi(A)CG$. Assume that the condition
\be\la{A.5}
\gotH^{ac}(A) = \clospa\bigl\{E^{ac}(\gd)\ran(C):\gd \in \kB(\bR)\bigr\}
\ee
is satisfied and let $K(\lambda) =
\frac{d}{d\gl}C^*E((-\infty,\gl))C$ and $\kH_\gl = \overline{\ran(K(\gl))}$ for a.e. $\lambda\in\dR$. 

Then $L^2(\bR,d\gl,\kH_\gl)$ is a  spectral representation
of $A^{ac}$ and the scattering matrix 
$\{S_{AB}(\gl)\}_{\gl \in \bR}$ of the scattering system $\{A,B\}$ has the representation
\begin{equation}\label{A.3}
S_{AB}(\gl) = I_{\kH_\gl} + 2\pi i (1 + \gl^2)^2
\sqrt{K(\gl)}Z(\gl)\sqrt{K(\gl)} \in[\kH_\gl]
\end{equation}
for a.e. $\gl \in \bR$, where 
\begin{equation}\label{A.2}
Z(\gl) = \frac{1}{\gl+i}Q^*Q +
\frac{\phi(\gl)}{(\gl+i)^2}G + Q_B(\gl+i0)
\end{equation}
and the limit $Q_B(\gl+i0) := \lim_{\varepsilon\to+0}Q^*(B-\lambda-i\varepsilon)^{-1}Q$ is taken in the Hilbert-Schmidt norm.
\et

\end{appendix}

\subsection*{Acknowledgment}

The second author thanks the Weierstrass Institute of Applied
Analysis and Stochastics in Berlin for financial support and hospitality.

 \def\cprime{$'$} \def\cprime{$'$} \def\cprime{$'$} \def\cprime{$'$}
  \def\lfhook#1{\setbox0=\hbox{#1}{\ooalign{\hidewidth
  \lower1.5ex\hbox{'}\hidewidth\crcr\unhbox0}}} \def\cprime{$'$}
  \def\cprime{$'$} \def\cprime{$'$}

\end{document}